\documentclass[journal,10pt]{IEEEtran}
\IEEEoverridecommandlockouts
\usepackage[utf8]{inputenc}
\usepackage{amsmath,amssymb,amsfonts}
\usepackage{amsthm}
\usepackage{algorithmic}
\usepackage{graphicx}
\usepackage{epstopdf}
\usepackage{textcomp}
\usepackage{xcolor}
\usepackage{bbm}
\usepackage[english]{babel}

\usepackage{caption}
\usepackage{subcaption}
\usepackage{floatrow}
\usepackage{wrapfig}

\newtheorem{lemma}{Lemma}
\newtheorem{theorem}{Theorem}
\newtheorem{remark}{Remark}

\newcommand{\Ex}{\mathbb{E}}

\newcommand{\R}{\mathcal{R}}
\newcommand{\N}{\mathcal{N}}

\DeclareMathOperator*{\argmax}{\arg\!\max}
\DeclareMathOperator*{\argmin}{\arg\!\min}

\usepackage{latexsym} 
\makeatletter
\def\endthebibliography{%
	\def\@noitemerr{\@latex@warning{Empty `thebibliography' environment}}%
	\endlist
}
\makeatother

\title{User Association in Dense mmWave Networks as Restless Bandits
}
\begin{document}
	\author{Santosh Kumar Singh, Vivek S. Borkar, and Gaurav S. Kasbekar}
	\maketitle
{\renewcommand{\thefootnote}{} \footnotetext{Copyright (c) 2015 IEEE. Personal use of this material is permitted. However, permission to use this material for any other purposes must be obtained from the IEEE by sending a request to pubs-permissions@ieee.org.
		}}
	
{\renewcommand{\thefootnote}{} \footnotetext{S.K. Singh, V.S. Borkar, and G.S. Kasbekar  are with the Department of Electrical Engineering, Indian Institute of Technology (IIT) Bombay, Mumbai, India. Their email addresses are santoshiitb@ee.iitb.ac.in, borkar@ee.iitb.ac.in, and gskasbekar@ee.iitb.ac.in, respectively. The work of SKS and GSK was supported in part by the project with code RD/0121-MEITY01-001. The work of VSB was supported in part by a S.\ S.\ Bhatnagar Fellowship from the Government of India. 
}}
	\begin{abstract}
We study the problem of user association, i.e., determining which base station (BS) a user should associate with, in a dense millimeter wave (mmWave) network.  	
In our system model, in each time slot, a user arrives with some probability in
a region with a relatively small geographical area served by a dense mmWave network. Our goal is to devise an association policy under which, in each time slot in which a user arrives, it is assigned to exactly one BS so as to
minimize the weighted average amount of time that users
spend in the system. The above problem is a restless multi-armed
bandit problem and is provably hard to solve. We prove that the problem is Whittle indexable, and based on this result, propose an association
policy under which an arriving user is associated with the BS having the smallest Whittle index. Using simulations, we show that our proposed policy outperforms several user association policies proposed in prior work.  	
\end{abstract}
	
	
\section{Introduction}
Recently, there has been an exponential increase in the volume of data traffic exchanged using wireless networks~\cite{Cisco} due to a proliferation  of data-hungry services with high Quality-of-Service (QoS) requirements. The traditionally used sub-6 GHz cellular bands, which are crowded and expensive, are unable to meet the ever increasing data volume and QoS requirements, despite the use of advanced techniques such as Massive Multiple Input Multiple Output (MIMO) and heterogeneous networking~\cite{Andrews14JSAC,Agiwal16COMST}. In contrast, ample un-utilized spectrum is available~\cite{Cudak13VTC} in millimeter wave  (mmWave) bands and it has the potential to provide multi-gigabit data rates~\cite{Niu15WNET}. Note that the behavior of the channel (medium) at sub-6 GHz and at mmWave frequencies is significantly different, e.g., in terms of attenuation, reflection and diffraction properties~\cite{Andrews14JSAC}. In particular, sub-6 GHz  waves experience lower attenuation with distance, higher reflection from surfaces, higher diffraction from edges and higher penetration loss from blocking objects in the environment than mmWaves~\cite{Niu15WNET,Zhao13ICC}. The above differences in the behavior of the channel (medium) result in differences in the optimal deployment of BSs, sizes and shapes of cells,  association and handover of user with BSs, etc. In particular, in sub-6 GHz networks,  often the cell size is approximately hexagonal or circular and the BS is placed at its center, and handovers typically occur at the boundaries of cells~\cite{Andrews14JSAC,Agiwal16COMST}. Also, in mmWave networks, the deployment of BSs needs to be denser, the nature of communication directional, and handovers do not necessarily occur at the boundaries of cells~\cite{Feng17JSAC,Yang18TCOMM}. The directional nature of communication, dense deployment of BSs, and short transmission range in mmWave networks pose several challenges~\cite{Niu15WNET} such as blockage, rare but heavy interference, frequent handovers, etc.,  and novel strategies are required  to deal with the above challenges~\cite{Sakaguchi17TIEICE}.  
	
The process of user association, i.e., determining which BS a given user should associate with, is crucial in both sub-6 GHz and mmWave wireless networks~\cite{Liu16COMST,Attiah20WNET}. Hence, user association problems have been extensively investigated in prior work, in different network scenarios, with different objective functions and constraints. 
The problem of user association in sub-6 GHz networks has been studied with the objective of maximizing throughput in~\cite{Xu2015JSYST,Feng17TVT,Ma17LWC,Sun18TVT,Dong19TVT,Khalili20LWC}, balancing load in~\cite{Kim11TNET,Ye13TWC,Alizadeh20GLOBECOM}, maximizing fairness in~\cite{Shen14JSAC}, maximizing energy efficiency in~\cite{Fang20TWC,Zarandi20TGCN} and optimizing network utility in~\cite{Bethanabhotla15TWC}. The problem of user association in mmWave networks has been studied with the objective of maximizing throughput in~\cite{Khan2019TCCN,Sana20TWC}, balancing load in~\cite{Zhou15TETC,Alizadeh19TWC,Liu19TCOMM,Khawam20TMC}, maximizing energy efficiency in~\cite{Zhang17JSAC}, maximizing line of sight (LoS) connectivity in~\cite{Soleimani18TCOMM}, optimizing BS deployment in~\cite{Zhang20TWC} and optimizing handovers in~\cite{Sun20TMC}. In prior work, most user association problems were formulated as constrained optimization problems-- in particular, as  combinatorial optimization problems in~\cite{Feng17TVT,Dong19TVT,Khalili20LWC,Zarandi20TGCN,Fang20TWC,Zhang17JSAC,Soleimani18TCOMM,Alizadeh19TWC,Liu19TCOMM,Khawam20TMC}, as non-convex optimization problems in~\cite{Khan2019TCCN,Sana20TWC} and as stochastic optimization problems in~\cite{Sun18TVT,Alizadeh20GLOBECOM,Sun20TMC}. Tools based on the gradient algorithm, Lagrangian method, game theory, machine learning, etc., were used to solve the above problems. 

In~\cite{Alizadeh20GLOBECOM}, the problem of user association in cellular heterogeneous networks (HetNets) was modeled using the multi-armed bandit framework and solved using tools from reinforcement learning. In~\cite{Sun20TMC}, the multi-armed bandit framework was used to model handovers in dense mmWave networks and an online learning algorithm for performing handovers was proposed using the empirical distribution of LoS blockage and post handover trajectories of users. However, none of the above works provided an index based association policy. In~\cite{Sun18TVT}, user association in cellular HetNets was modeled as a \textit{restless multi-armed bandit problem} \footnote{A collection of two or more controlled stochastic processes with two controls, say-- active and passive-- and with discrete state space, is said to be restless bandits, if each process and in each state, upon application of any of the two controls, changes its state, but with different probability law. A reward (cost) is obtained (incurred) upon application of control for each process, and depends on triplet-- current state, action, and next state. Note that classical multi-arm bandits is different from the restless bandits in sense that in the classical multi-arm bandits, the processes for which  passive control is applied does not change its state and provides (incur) zero reward (cost). The goal of restless bandits problems is to maximize (minimize) long run average/ discounted reward (cost) given constraints that in each slot exactly some fixed number of processes to be remain active~\cite{Whittle88restless}.} and it was proved that the modeled problem is PSPACE hard~\cite{Papadimitriou94complexity}. The authors derived the association priority index (which is different from the Whittle index\footnote{A brief description about Whittle Index is provided in Section \ref{Background on Whittle index}.}~\cite{Whittle88restless}) for small cell BSs using the primal-dual index heuristic algorithm after relaxing the hard per stage constraint. However, to the best of our knowledge, \emph{Whittle index~\cite{Whittle88restless} has not been used for solving the association problem in prior work}. This is the space in which we contribute in this paper. We have formulated the association problem in dense mmWave networks as  a \textit{restless multi-arm bandit problem} and provided a Whittle index based user association policy. Note that the Whittle index was introduced in~\cite{Whittle88restless} and has been successfully used for solving problems in a variety of applications~\cite{raghunathan2008index,Argon09PEIS,Liu10TIT,Avrachenkov16TCNS,Borkar17TCNS,Borkar2017AOPR,Hsu18ISIT,Xu19TVT,Tripathi19ALLERTON,Avrachenkov19ALLERTON,Wang19TAC,Sombabu20COMSNETS,Wu20TCNS,Chen21Access}.

In this paper, we consider the user association problem in a dense mmWave network serving a region with a relatively small geographical area such as a seminar hall, bus stop, etc. Time is divided into slots of equal duration and in each slot, a user arrives with some probability. Our goal is to devise an association policy under which, in each time slot in which a user arrives, it is assigned to exactly one BS so as to minimize the weighted average amount of time that users spend in the system. The above problem is a restless multi-armed bandit problem and is provably hard to solve~\cite{Papadimitriou94complexity}. Using an idea of Whittle~\cite{Whittle88restless}, we relax the exact constraint, in which an arriving user needs to be associated with \emph{exactly one BS}, to a time-averaged constraint, in which, an arriving user is associated with one BS \emph{on average}. The use of this relaxation and the standard Lagrange multiplier technique lead us to a set of decoupled controlled Markov chains or Markov decision processes (MDP). We prove the Whittle indexability of each MDP and then establish the indexabilty of the original problem; based on this result, we propose an association policy, in which an arriving user is associated with the BS having the smallest Whittle index.
Our contribution is non-trivial since establishing the Whittle indexability of restless multi-armed bandit problems is intractable in many scenarios. We compare our proposed Whittle index based association policy with the load based, Signal to Noise Ratio (SNR) based, throughput based, mixed, and random association policies~\cite{Gupta21Access} (see Section~\ref{Sec 7} for descriptions of these policies) via detailed simulations and show that our policy outperforms the other policies in all scenarios.

The rest of this paper is organized as follows. Section~\ref{sec_related_work} presents relevant prior work on user association. Section~\ref{section 2} describes the system model and problem formulation. Sections~\ref{Sec 3} and~\ref{Sec 4} establish the threshold nature of the optimal policy. Section~\ref{Sec 5} establishes the Whittle indexability of the considered problem. Section~\ref{Sec 6} describes a scheme for computation of the Whittle index. Section \ref{Section: Application of our results to sub-6 GHz networks} provides application of our results to sub-6 GHz networks. Section~\ref{Sec 7} describes other user association policies used for comparison with our proposed policy and Section~\ref{Sec 8} presents simulation results. Finally, Section~\ref{SC:conclusions} concludes this paper. 

\section{Related Work} \label{sec_related_work}
We review prior works on user association in sub-6 GHz and mmWave networks in Sections~\ref{subA_related work} and~\ref{subB_related work}, respectively. We explain the differences between our work and prior work in Section~\ref{SSC:related:work:differences}.

\subsection{Association in Sub-6 GHz Networks} \label{subA_related work}
		
The alpha-optimal user association policy was proposed in~\cite{Kim11TNET} to adapt to traffic load heterogeneity across BSs in homogeneous networks. This policy becomes optimal in different contexts such as throughput-optimization, delay-optimization,  etc., for different values of alpha. In HetNets, biasing at users for small cell BSs plays a significant role in balancing load and maximizing network throughput. In~\cite{Jo12TWC},  the authors evaluated the effect of biasing on coverage probability and in~\cite{Ye13TWC}, the authors provided a distributed algorithm for finding a load balancing optimal bias in multi-tier HetNets. In~\cite{Shen14JSAC}, a pricing based distributed algorithm was proposed for HetNets with massive MIMO enabled BSs, with the aim of maximizing fairness, considering two situations-- the channel experiences flat fading and frequency selective fading. In~\cite{Bethanabhotla15TWC}, a non-cooperative game based distributed algorithm was proposed for HetNets with BSs with varying numbers of antennas and transmission power capabilities. Each user acts as a player and selfishly chooses a BS based on the utility it gets upon association; BSs allocate their resources based on a local resource allocation rule. In~\cite{Xu2015JSYST}, an optimal centralized algorithm was proposed with the aim of maximizing the network throughput in HetNets, after establishing the unimodularity of the considered problem.

In~\cite{Ma17LWC}, the user cluster-BS association problem was formulated as a combinatorial optimization problem with the aim of maximizing the throughput in massive MIMO enabled HetNets and a low complexity algorithm was proposed to solve it.  User association was formulated as a mixed integer nonlinear programming problem in~\cite{Feng17TVT,Dong19TVT,Khalili20LWC} with the aim of maximizing throughput and in~\cite{Zarandi20TGCN,Fang20TWC}, with the aim of maximizing energy efficiency.  In~\cite{Feng17TVT}, the authors considered two cases-- the problem of finding an interference nulling schedule given user association and the problem of jointly finding an interference nulling schedule and user association. For the first case, an optimal solution based on the cutting plane approach was provided and for the second case,  a distributed scheme was proposed to poly match users with BSs.  In~\cite{Dong19TVT},  a low complexity algorithm based on a three step Gaussian belief propagation (GaBP) distributed solver was proposed. In~\cite{Khalili20LWC}, considering a constraint on cross tier interference, an iterative algorithm was proposed for joint optimization of user association, carrier allocation, antenna selection, and power allocation, using tools from majorization-minimization theory and the augmented Lagrangian method.  In~\cite{Zarandi20TGCN}, the original problem was reformulated as the problem of simultaneous maximization of throughput and minimization of power and an epsilon-method based algorithm was proposed to solve the reformulated problem. In~\cite{Fang20TWC}, the original problem was decomposed into two sub-problems-- association and power allocation-- and an iterative algorithm was provided for the cases with and without co-channel interference.

The user association problem was formulated as a stochastic optimization problem    with the aim of maximizing throughput (respectively, balancing load) in~\cite{Sun18TVT} (respectively, in~\cite{Alizadeh20GLOBECOM}). In~\cite{Sun18TVT}, the problem was shown to be PSPACE hard, and solved in two steps. In the first step, the restless multi-armed bandit framework was used to derive an association priority index (which is different from the Whittle index~\cite{Whittle88restless}) for small cell BSs and in the second step, the proposed algorithm chose the BS with the smallest association priority index from a set of small cell BSs chosen based on the signal to interference plus noise ratio (SINR). In~\cite{Alizadeh20GLOBECOM}, the authors proposed a centralized and a semi-distributed online algorithm, using the multi-armed bandit technique, for performing load balancing and achieving high spectral efficiency. 

However, all the above works address the user association problem in sub-6 GHz networks. In contrast, we address the problem of user association in mmWave networks.

	\subsection{Association in mmWave Networks} \label{subB_related work}

In~\cite{Zhou15TETC}, the user association problem was formulated as a mixed integer linear programming (MILP) problem with the aim of balancing the load in an mmWave network with one macro and several femto cell BSs, by considering joint optimization of association and scheduling. In~\cite{Khan2019TCCN, Sana20TWC}, the user association problem was formulated as a non-convex optimization problem with the aim of maximizing throughput. In~\cite{Khan2019TCCN}, using deep reinforcement learning and the actor critic algorithm, a low complexity algorithm was proposed, which approximates the solution of the original optimization problem. In~\cite{Sana20TWC}, a low complexity, scalable, and flexible algorithm based on multi-agent reinforcement learning was proposed for user association, in which users act as independent agents and adapt their actions based on local information of network states. In~\cite{Sun20TMC}, the user association problem was formulated as  a stochastic optimization problem with the aim of optimizing handovers. The authors proposed two handover mechanisms focusing on spatial and space-time contexts, respectively. The proposed algorithms learn online in a multi-armed bandit framework by exploiting the user's post trajectory distribution. The proposed algorithms do not assume prior knowledge of the user's mobility and environment.
	
The user association problem was formulated as a mixed integer non linear programming (MINLP) problem in~\cite{Zhang17JSAC} with the aim of maximizing energy efficiency, in~\cite{Soleimani18TCOMM} with the aim of maximizing the LoS connectivity and in~\cite{Alizadeh19TWC,Liu19TCOMM,Khawam20TMC} with the aim of balancing the load. In~\cite{Zhang17JSAC}, the formulated problem considered load balancing constraints, a limit on cross tier interference and user QoS requirements; an iterative gradient based algorithm was proposed for user association and power allocation. In~\cite{Soleimani18TCOMM}, the authors used the  difference of two convex programming problems to solve the problem obtained after relaxation of binary variables and proposed  a near-optimal polynomial-time algorithm to assign femto-cell users to femto-cell BSs. In~\cite{Alizadeh19TWC}, a near-optimal polynomial-time worst connection swapping algorithm was proposed and shown to outperform other generic algorithms used to solve combinatorial optimization problems in terms of both accuracy and speed.  In~\cite{Liu19TCOMM}, the authors designed an iterative algorithm for joint user association and power allocation using the Lagrange dual decomposition and Newton-Raphson methods for a single-band access scheme and obtained a near-optimal solution based on the Markov approximation framework for a multi-band access scheme.  In~\cite{Khawam20TMC}, the original problem was reformulated as a non-cooperative game and an efficient distributed solution was provided. 
	
\subsection{Differences Between Our Work and Prior Work}
\label{SSC:related:work:differences}
In most prior works, the user association problem was formulated as a constrained optimization problem, stochastic optimization problem, game theoretic problem, etc., and only a few works, viz.,~\cite{Sun18TVT,Alizadeh20GLOBECOM,Sun20TMC} formulated the user association problem using the multi-armed bandit framework. Out of the latter works, only in~\cite{Sun18TVT}, an index based user association policy was provided. However, the index used in~\cite{Sun18TVT} is different from the Whittle index. To the best of our knowledge, our work is the first to use the Whittle index, which has been successfully used for solving problems in a variety of applications, for solving the user association problem. 

\begin{remark}
Some of the proofs in this paper are similar to those in~\cite{Borkar2017AOPR}. 
However, there are several differences between the model in this paper and that in~\cite{Borkar2017AOPR}. For example, this paper considers the problem of user association in mmWave networks, whereas~\cite{Borkar2017AOPR} considers the  problem of allocating jobs to processors in an egalitarian processor sharing setup. Also, departures from a mmWave base station are assumed to follow a Bernoulli process in this paper, whereas departures from a processor in~\cite{Borkar2017AOPR} are assumed to follow a  Binomial process. Due to the above differences, the analyses in this paper and in~\cite{Borkar2017AOPR} are significantly different. 
\end{remark}

\section{Model, Problem Formulation, and Background} \label{section 2}
\subsection{Model and Problem Formulation}
Consider a wireless network with $K$ mmWave base stations (mBSs) 
serving a small region (e.g., a seminar hall, a bus stop, etc.). Time is divided into slots of equal duration; also, a slot $n \in \{0, 1, 2, \cdots\}$ is considered as the duration of time from time instant $n$ to $n+1$. In each slot, a user arrives (respectively, no user arrives) into the region with probability (w.p.) $p$ (respectively, $1-p$), where $0 < p < 1$. An arriving user gets rate $R_i$ when it gets associated with mBS $i \in \{1, \cdots,K\}$. We assume that all the users associated with mBS $i$ get the same rate $R_i$.  Note that $R_i$ has one of the usual units of rate such as Gbps or Mbps. This assumption models a scenario in which all the users arrive into a relatively small area, due to which, the channel quality and hence rate to mBS $i$ is the same for every user. To facilitate mathematical analysis, from this point onwards, with some abuse of notation, we refer to the normalized rate of each user from a mBS $i$ by $r_i=\frac{R_i}{(\max_{j \in \{1,\cdots,K\}}R_j)+\delta}$, where  $\delta>0$ is a constant. Note that $r_i \in (0,1)$  and is unit-less.  Fig.~\ref{system_model} illustrates the system model.

		\begin{figure}[h!]
		\centering
		\includegraphics[scale=1.2]{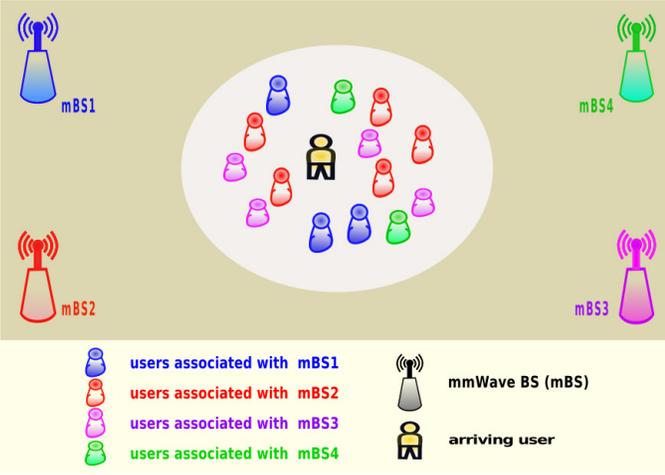}
		\caption{The figure shows an example of the considered wireless network with four mBSs.}
		\label{system_model}
	\end{figure}

Let:
	\begin{align*}
		\zeta_{n+1}=\begin{cases}
			1,&\text{if a user arrives in slot $n$},\\
			0,& \mbox{else}.
		\end{cases}
	\end{align*}
Note that $\zeta_{n+1}$ is a Bernoulli random variable with parameter $p$. Let:
	\begin{align*}
		u_{n}^i=\begin{cases}
			1, &\text{if mBS $i$ admits an arrival in slot $n$,}\\
			0, & \mbox{else.}  
		\end{cases}
	\end{align*}
When an arrival occurs in a slot, it should be assigned to exactly one mBS. Thus, we have the following constraint:
	\begin{align*}
		\sum_{i=1}^{K}u_{n}^i=1 \qquad \forall n.
	\end{align*}
The decision-- on which mBS an arriving user is to be assigned to-- is made taking into account the number of users associated with every mBS. We assume that the mBSs are connected to each other (e.g., they may all be connected to a controller or there may be pair-wise communication links among them); so the above information about numbers of associated users can be exchanged among the BSs.    
	
Let $X^i_n$ denote the number of users associated with mBS $i$ at the beginning of slot $n$. If $X^i_{n} \ge 1$, then a user  (respectively, no user) departs from the queue of mBS $i$ w.p. $r_i$ (respectively, $1-r_i$) in slot $n$. Note that the probability that a user departs (respectively, no user departs) from the queue of mBS $i$ in slot $n$ is independent of $X^i_n$. We use this simple model for the following reasons: The total rate at which the users associated with mBS $i$ in slot $n$ are served is approximately independent of $X^i_n$-- the available bandwidth is shared among the $X^i_n$ users. Also, larger the value of $X^i_n$, more the candidates for departure; on the other hand, smaller the value of $X^i_n$, higher the rate at which each of the users is served.
Let:
	\begin{align*}
		\gamma_{n+1}^i=\begin{cases}
			1, &\text{if a user departs from the queue of mBS $i$ in}\\
			& \mbox{slot $n$,}\\
			0, & \mbox{else}.
		\end{cases}
	\end{align*}
Then $\gamma_{n+1}^i$ is a Bernoulli random variable with parameter $r_i$.

\begin{remark}
In our model, the duration of a time slot is small, and so the probability of two or more users departing from the queue of a mBS in a single time slot is low. Hence, for simplicity, we have assumed that in each slot, at most one user may depart from the queue of a mBS.  
\end{remark}

We are interested in devising a non-anticipating admissible policy, i.e., $\forall n$, given $\{X^i_0; \gamma^i_m,\zeta_m, \,\, m \le n ; u^i_m, m<n \}$, the action  $u_{n}^i$ is to be conditionally independent of $\gamma_m^i, \zeta_m,\,\, 1\le i\le K,\,\, m> n$. Without loss of generality, we assume that $1>r_1\ge \cdots\ge r_K> 0$. Also, to ensure stability of the controlled queues obtained after decoupling of the problem in (\ref{Eq4}) below, we assume that $\frac{p}{1-p} < r_K$. 
	
The state of the queue at mBS $i$ is updated at time instant $n+1$ (i.e., at the end of slot $n$) as:
	\begin{align}
		X^i_{n+1}=(X^i_{n}+\zeta_{n+1} u_{n}^i -\gamma_{n+1}^i)^+, \label{Eq1}
	\end{align}
where $x^+ := \max(x,0)$. A cost $C_i > 0$ per slot per user is incurred at mBS $i$. The cost $C_i$ can be interpreted as the QoS provided by mBS $i$ to users associated with it in terms of the average delay provided by it. The different values of cost at different BSs model heterogeneity across mBSs, e.g., in the sense of different backhaul capacity and/ or different numbers of RF-chain.  The total cost experienced by the mBSs in the system in slot $n$ is given by:
\begin{align*}
	\sum_{i=1}^{K}C_i X^i_n. 
\end{align*}

Our objective is to choose $\{u^i_n\}$, $i \in \{1,\cdots,K\}$, $n \in \{0,1,\cdots\}$, to minimize the long-run expected average cost incurred at the mBSs in the network. Hence, we seek to solve the following problem:
	\begin{align}
		\mbox{minimize}&\limsup_{N\uparrow\infty} \Ex\Bigg[\frac{1}{N} 
		\sum_{n=0}^{N-1} \Bigg(\sum_{i=1}^{K}C_i X^i_n \Bigg) \Bigg] \label{Eq2}\\
		\mbox{s.t.}\qquad  &\sum_{i=1}^{K}u_n^i=1, \qquad \forall n. \nonumber
	\end{align}
Note that the cost in \eqref{Eq2} is the weighted average amount of time that users spend in the system; minimizing this cost ensures that the traffic of users is served fast on average. 

\subsection{Background on Whittle Index} \label{Background on Whittle index}	
The constrained problem in (\ref{Eq2}) is a restless bandit problem with a hard per-stage constraint and obtaining an optimal solution for it is provably hard~\cite{Papadimitriou94complexity}. Whittle in~\cite{Whittle88restless} proposed  that the hard per-stage constraint be relaxed to an average constraint to obtain the relax constraint problem which can provides index based heuristics as a solution satisfying the hard constraint. Note that the optimality of the Whittle index based heuristic, obtained from the optimal solution of the relaxed problem, has been proved in a very few well designed simple cases under suitable assumptions~\cite{Liu10TIT,Maatouk20TWC,Kriouile21ArXiv}. However, in many cases, it might be possible that the optimal solution obtained from the problem with the relaxed constraint may not be a feasible solution for the original problem; even if it is feasible, it may not be an optimal solution for it.  In general, the Whittle index based heuristic is known to be optimal in an asymptotic sense in the infinitely many bandits limit~\cite{weber1990index}.  
The hard per-stage constraint is relaxed to the following average per-stage constraint:
	\begin{align}
		\limsup_{N\uparrow \infty}\frac{1}{N} \sum_{n=0}^{N-1}\sum_{i=1}^{K} \Ex[u_n^i]=1. \label{Eq3}
	\end{align}
The above constraint has the same form as the objective in (\ref{Eq2}), which  paves the way towards a solution obtained by using the standard Lagrange  multiplier formulation. The problem with the relaxed constraint then gets converted into the following unconstrained problem:
	\begin{align}
	\mbox{minimize}	& \limsup_{N\uparrow \infty}\frac{1}{N} \sum_{n=0}^{N-1}\sum_{i=1}^{K} \Ex[\digamma_i(X^i_n,u_n^i)],\label{Eq4}\\
		\mbox{where} \qquad & \digamma_i(x,u)=C_i x+(1-u)\lambda, \nonumber
	\end{align}
and $\lambda$ is a Lagrange multiplier. The masterstroke of Whittle in~\cite{Whittle88restless} was to interpret the Lagrange multiplier $\lambda$ as a  subsidy for a reward-maximization problem. Our problem here is a cost-minimization problem; hence, we choose the above specific form of the cost function and view the Lagrange multiplier as a tax or negative subsidy in the  sense of Whittle~\cite{Whittle88restless}. 
	
Given $\lambda$, the problem in (\ref{Eq4}) gets decoupled into separate controlled chains or MDPs, one controlled chain corresponding to each mBS. To devise a policy based on the Whittle index, we first need to prove that the original problem is Whittle indexable~\cite{Whittle88restless}. If for each decoupled chain and for all sets of parameter values $\{C_i,r_i,p\}$, the set of states for which it is optimal for an mBS to not accept an arrival decreases monotonically from the whole state space to the empty set as the tax $\lambda$ increases from $-\infty$ to $\infty$, then the original problem is said to be Whittle indexable~\cite{Whittle88restless}. The Whittle index of a chain for a state is the value of the tax for which, under the optimal policy, the mBS is indifferent between the two actions-- accepting and not accepting an arrival. In each slot, the Whittle index based policy for the original problem is: the mBS with the smallest Whittle index accepts the arrival. Note that while the Whittle policy is arrived at via a relaxation of the original per-stage constraints, it does satisfy the original constraints.
	
\section{Optimal Policy and Value Function } \label{Sec 3}
This section proves two results (Lemmas~\ref{Lemma1} and~\ref{Lemma3}), which characterize the optimal stationary policy and provide an equation satisfied by the value function. 
	
	\begin{lemma} \label{Lemma1}
		For $X_n^i, 1 \leq i \leq K,$ the function $\psi(x)=e^{\sigma x}$, where $\sigma>0$ is such that $e^{\sigma}<\frac{r_K(1-p)}{p}$, acts as a Lyapunov function  satisfying: under any stationary policy,
		$$\Ex\big[\psi(X^i_{n+1})-\psi(X^i_n)|X^i_n\big]\le -\delta \psi(X^i_n),$$  $\mbox{ for } X^i_n \ne 0$ \mbox{ and }
		for some $\delta>0$.
	\end{lemma}
	\begin{proof} We drop the superscript $i$ for convenience. 
		For $X_n \neq 0$,
\begin{align*}
			&\Ex\big[ \psi(X_{n+1})-\psi(X_n)|X_n \big] \\
			&\le \Bigl\{ p \big( e^{\sigma}-1   \big)+ (1-p) \Ex \Big( e^{-\sigma} \gamma_{n+1}-1 | X_n \Big) \Bigr\} \psi(X_n) \\
			&= \Bigl\{ p\big(e^{\sigma}-1\big)+(1-p)\Pr(\gamma_{n+1} = 1| X_n)\big(e^{-\sigma}-1\big) \Bigr\} \psi(X_n)\\
			&\le \big\{p \big(e^{\sigma}-1\big)+(1-p)r_K\big( e^{-\sigma}-1 \big)\big\} \psi(X_n)
\end{align*}
The result follows. 
	\end{proof}
 
Lemma~\ref{Lemma1} provides a guarantee that the controlled chain eventually hits the state zero with probability $1$ regardless of the chosen stationary control policy and the initial state. Lemma~\ref{Lemma1} is used in establishing the results stated in Lemma~\ref{Lemma3}, which provide insight into the optimal stationary control policy.

Recall that our original problem of minimizing the long-run average cost under the  per-stage hard constraint gets converted into separate control problems of minimizing the long-run average cost for each mBS given the tax $\lambda$. Since the proof of the result that the decoupled problem corresponding to mBS $i$ is Whittle indexable is the same for each mBS $i$, henceforth we drop the index $i$ corresponding to the mBS for simplicity. The dynamic programming equation satisfied by the value function of the individual problem is: 
	\begin{eqnarray}
		V(x) &=& Cx-\rho+\min\Big([(1-p)(1-r)+p r]V(x) \nonumber\\
		& & +(1-p)r V((x-1)^+)+p (1- r )V(x+1); \nonumber \\ 
& & \lambda+(1-r)V(x)+rV((x-1)^+) \Big). \label{Eq8} 
	\end{eqnarray}
The rest of this section sketches the derivation of (\ref{Eq8}). 	

Let $0<\beta<1$. Under the stationary control policy $\pi$, the infinite horizon $\beta$-discounted cost for the controlled process starting in state $x$ is:
	\begin{align*}
		I^{\beta}(x,\pi):=\Ex\Big[\sum_{n=0}^{\infty}\beta^n (C X_n +(1-u_n)\lambda)|X_0=x \Big].
	\end{align*}
The value function for the above infinite horizon $\beta$-discounted problem will be the minimum over all stationary control policies and is given by:
	\begin{align*}
		V^{\beta}(x)=\min_{\pi} I^{\beta}(x,\pi).
	\end{align*}
	
Let $p_{\cdot|\cdot}(u)$ be the transition probability of the controlled chain. Then the value function satisfies the following dynamic programming equation:
	\begin{align*}
		V^{\beta}(x)=\min_{u}\Big[Cx+(1-u)\lambda+\beta \sum_{y} p_{y|x}(u)V^{\beta}(y) \Big].
	\end{align*}
Let $\bar{V}_{\beta}(\cdot)=V^{\beta}(\cdot)-V^{\beta}(0)$. Then, $\bar{V}_{\beta}(\cdot)$ satisfies:
	\begin{align*}
		\bar{V}^{\beta}(x)&=\min_{u}\Big[Cx+(1-u)\lambda-(1-\beta)V^{\beta}(0)\\
		&+\beta \sum_{y} p_{y|x}(u) \bar{V}^{\beta}(y) \Big].
	\end{align*}
	
We now state a lemma, which will be used to prove some results in the following sections. 
	
	\begin{lemma} \label{Lemma3}
		$\lim_{\beta \uparrow 1}\bar{V}^{\beta}=V$ and $\lim_{\beta \uparrow 1}(1-\beta)V^{\beta}(0)=\rho$, where $(V, \rho)$ satisfy (\ref{Eq8}). Furthermore,
		$\rho$ is uniquely characterized as the optimal cost $\rho(\lambda)$ and $V$ is rendered unique on states that are positive recurrent under an optimal policy under the additional condition $V(0) = 0 $. Finally, the argmin of the RHS of (\ref{Eq8}) yields the optimal choice of $u$ for the state $x$. 
	\end{lemma}
	\begin{proof}
		The proof uses Lemma~\ref{Lemma1} and follows by an argument similar to that used to prove Lemma 4 on p. 7 of~\cite{Borkar2017AOPR}. We omit the details for brevity. 
	\end{proof}

\section{Threshold Nature of Optimal Policy} \label{Sec 4}
In this section, we propose a   stationary threshold policy using structural properties of the value function. For this purpose, we relax the state space to $[0,\infty)$ and the control space to $[0,1]$. For this relaxation, the following structural property holds:
\begin{lemma}\label{Lemma Vind}
	$V$ is monotone increasing and  has non-decreasing differences, i.e., if $y > 0$ and $x > x^{\prime}$, then:
	$$V(x + y) - V(x) \ge V(x^{\prime} + y) - V(x^{\prime}).$$
\end{lemma}
\begin{proof}
	Since convexity of a function implies non-decreasing differences,  to prove that $V$ has non-decreasing differences, it suffices to prove that it is convex. To prove that $V$ is convex, it is sufficient to prove that the value function of the infinite horizon $\beta-$discounted problem is convex for all $\beta$. This is because the pointwise limit of a sequence of convex functions is also convex and we can choose  $\beta_n \uparrow 1$ such that $V^{\beta_n}(\cdot) -  V^{\beta_n}(0) \to  V(\cdot)$. Since the infinite horizon $\beta-$discounted problem is a limiting case of the finite horizon $\beta-$discounted problem, it suffices to prove that the finite horizon $\beta-$discounted value function is convex.
	We will prove the convexity of the value function of the finite horizon $\beta-$discounted problem using an induction argument for the continuous state space of positive reals and continuous action space $[0,1]$ for the control action that allows the admission of a fraction $u$ of the arriving user.   
	
	Let  $P_a(\cdot)$ and $P_d(\cdot)$  be the distributions of the arrival random variable, $\zeta$, and departure random variable, $\gamma$, respectively. Note that both the arrival and departure random variables are state independent Bernoulli processes. Now  consider the dynamic programming equation for the $n-$step finite horizon $\beta-$discounted problem for a continuous state space $[0,\infty)$ and continuous action space $[0,1]$:
	\begin{align}
		V_m^{\beta}(x)&=\min_{u}\big[ Cx \ + \ ( 1 - u)\lambda \ + \nonumber \\
		&\beta \int V_{m-1}^{\beta}(x-\gamma+u \zeta) P_d(d\gamma)P_a(d\zeta)\big], \ \, 0<m\le n,  \label{Eq12}
	\end{align}
	with $V_0^{\beta}(x)=Cx,\, x\ge 0$. Define for $n > 0$:
	\begin{align}
		f_n^{\beta}(x,u)&=\big[ \digamma(x,u) \nonumber \\
		&+\beta \int V_{n-1}^{\beta}(x-\gamma+u \zeta)  P_d(d\gamma)P_a(d\zeta)  \big]. \label{Eq 8}
	\end{align}
	$V_0^{\beta}(x) = C x$ is a convex function. Assume that $V_{n-1}^{\beta}(x)$ is a convex function. Let $u_1$ and $u_2$ be the minimizers of the RHS of (\ref{Eq12})  at points $x_1$ and $x_2$, respectively, where $x_1>x_2$.
	Then we have:
	\begin{align*}
		&V_n^{\beta}(x_i)=f_n^{\beta}(x_i,u_i), \quad i \in \{1,2\}.
	\end{align*}
	Now,
	\begin{align*}
		&\alpha V_n^{\beta}(x_1)+(1-\alpha)V_n^{\beta}(x_2) \\
		&=\alpha \digamma(x_1,u_1)+(1-\alpha)\digamma(x_2,u_2)\\
		&+\beta \int \big[ \alpha V_{n-1}^{\beta}(x_1-\gamma+u_1 \zeta) \ + \\
		& (1-\alpha) V_{n-1}^{\beta}(x_2-\gamma+u_2 \zeta) \big] P_d(d\gamma)P_a(d\zeta)\\
		&\ge f_n^{\beta}(\alpha x_1+(1-\alpha)x_2,\alpha u_1 + (1-\alpha)u_2)\\
		&\ge  V_n^{\beta}(\alpha x_1+(1-\alpha) x_2)
	\end{align*}
	The first inequality holds by convexity of $f_n^\beta$ and the second inequality follows from the definition of $V^\beta_n$. This proves the convexity of $V^\beta_n$, from which the convexity of $V^\beta$ and therefore of $V$ follows by limiting arguments as already described. Monotone increase can also be proved by an analogous induction argument for finite horizon discounted problem followed by the infinite time and vanishing discount limits, in that order. \end{proof}

The dynamic programming equation (\ref{Eq8}) for this continuous state-action space formulation can be rewritten as 
\begin{eqnarray}
	V(x) &=& Cx-\rho+\min_{u\in[0,1]}\Big((1-p)(1-r)V(x) + \nonumber \\
	&& p(1-r)V(x+u) + (1-p)rV((x-1)^+) +\nonumber \\
	&&  prV((x+u-1)^+)+(1-\lambda)u\Big). \label{Eq8b} 
\end{eqnarray}
This involves minimization over $u \in [0, 1]$ of a function of the form $G(x, u) := F(x+ u) - \lambda u$, where $F$ is convex increasing.
Suppose this has a unique minimizer $u^* \in [0, 1]$. (The non-unique case can also be handled  by a suitable modification of what follows.) Since $F$ is convex, its right derivative $F'_+$ and left derivative $F'_-$ are defined except at most countably many points,  are monotone increasing, with $F'_+(x) \geq F'_-(x) \ \forall \ x \in \R^+$. 
Then we must have
$$\lambda \in \partial F(x + u^*) = [F'_-(x + u^*),   F'_+(x + u^*)].$$
Suppose $x + u^* \in [n, n+1]$ for some $n \geq 0$. Then for all  $m < n$, the function $u \mapsto G(y,u)$ is minimized at the $m+1$ (corresponding to $u = 1$). Similarly, for all $m \geq n + 1$, $u \mapsto G(y,u)$ is minimized at the $m$ (corresponding to $u = 0$). Thus  restricted to the original state space $S$, the optimal choice will lead to the next state that is  also in $S$ for every $m \in S$ except possibly for $m = n$, where if $u^* \in (0,1)$, it will take it to a point in $\R^+\backslash S$. If at $n$, the minimization were over $\{0,1\}$, the minimum would have been attained at $u = 0$, corresponding to $n$, if $F(n) \leq F(n+1) - \lambda$, and at $u^* = 1$ otherwise. If we opt for restricting the control to $\{0,1\}$, implying a suboptimal decision for at most one state, viz., $n$, we still get a threshold policy. We now work with this policy to show  as before that it satisfies the condition for Whittle indexability. Most importantly, once the final Whittle indices are derived, the Whittle policy chooses the active bandits accordingly, i.e., by picking a single arm of the bandit as dictated by the order of the indices for the current state profile. Then  the dynamics is very much within the original paradigm of state-action spaces $S$ and $U$.

The threshold nature together with the stability of the optimal policy exactly characterizes the form of the communicating class. In particular, it says that $\{0, \cdots, t+1\}$ will be a communicating class under the threshold policy with threshold $t$. By Lemma~\ref{Lemma1}, state $0$ is eventually reached by the process from all other states. This implies that at most one communicating class, possibly with some transient states, can exist. Since the process under the optimal threshold policy is stable, at least  one communicating class exists. Thus exactly one communicating  class, possibly with some transient states,  exists (unichain property) and it is of the form $\{0, \cdots, t+1\}$, because under the threshold policy with threshold $t$, the set of states $\{0,\cdots,t\}$ (respectively, $\{t+1,\cdots,\infty\}$) is the set for which mBS admits (respectively, does not admit) the arrival. For each $\lambda$, we get the optimal threshold policy for each decoupled process, which implies that we get sets of stationary threshold policies parameterized by $\lambda$.

\section{Whittle Indexability} \label{Sec 5}
In this section, we first prove a sequence of lemmas and then prove the Whittle indexability of the problem.

\begin{lemma} \label{Lemma 6} Let $\mu_t$ be the stationary distribution under the threshold policy with threshold $t$. Then $\sum_{j=0}^{t}\mu_t(j)$ is an increasing function of $t$.
\end{lemma}
\begin{proof}
	This can be shown using the idea of stochastic dominance of Markov chains. The proof is similar to that of Lemma 8 on p. 11 of~\cite{Borkar2017AOPR} and is omitted for brevity. 
\end{proof}
\begin{lemma} \label{Lemma 7}
	Suppose $g:\R \times \N  \rightarrow \R$ is submodular, i.e., $\forall \,\,\lambda_2<\lambda_1 \mbox{ and }  x_2<x_1 $, 
	\begin{align*}
		g(\lambda_1,x_2)+g(\lambda_2, x_1)\ge g(\lambda_1,x_1)+g(\lambda_2,x_2),
	\end{align*}
	and $x(\lambda):=\inf\{x^{\star}:g(\lambda,x^{\star})
	\le g(\lambda,x) \,\, \forall x \}$. Then $x(\lambda)$ is a non-decreasing function of $\lambda$.
\end{lemma}
\begin{proof}
	The proof follows from the discussion on p. 258 in Section 10.2 of~\cite{sundaram1996first}.
\end{proof}
\begin{lemma} \label{Lemma 8}
	Denote the stationary average cost under tax $\lambda$ and the threshold policy with threshold $t$ as $$g(\lambda, t)=C\sum_{j=0}^{\infty}j \mu_t(j)+\lambda\sum_{j=t+1}^{\infty}\mu_t(j).$$ Then $g$ is submodular.
\end{lemma}
\begin{proof}
	To prove that $g$ is submodular, we need to prove that:
	\begin{align*}
		&g(\lambda_1,t_2)+g(\lambda_2, t_1)\ge g(\lambda_1,t_1)+g(\lambda_2,t_2), \\
		& \qquad \qquad \qquad \qquad \qquad \forall \ \lambda_2<\lambda_1
		\mbox{ and } t_2<t_1.
	\end{align*}
	It is easy to see that the above inequality reduces to:
	\begin{align*}
		&\lambda_1 \sum_{i=0}^{t_2} \mu_{t_2}(i)+\lambda_2 \sum_{i=0}^{t_1} \mu_{t_1}(i)\le\lambda_1 \sum_{i=0}^{t_1} \mu_{t_1}(i)+\lambda_2 \sum_{i=0}^{t_2} \mu_{t_2}(i), \\
		& \qquad \qquad \qquad \qquad \qquad\qquad \qquad \quad  \forall \ \lambda_2<\lambda_1
		\mbox{ and } t_2<t_1.\\
		&\Longleftrightarrow  \sum_{i=0}^{t_2} \mu_{t_2}(i)\le \sum_{i=0}^{t_1} \mu_{t_1}(i), \,\,\qquad \forall \ t_2<t_1.
	\end{align*}
	By Lemma~\ref{Lemma 6}, the above inequality holds. Hence $g$ is submodular.
\end{proof}
At this point, we have all the ingredients needed to establish Whittle indexability. 

\begin{theorem}
	This problem is Whittle indexable.
\end{theorem}
\begin{proof}
	By the unichain property, there exists a unique stationary distribution under any stationary policy. Let $\mu$ be the unique stationary distribution and $\mathcal{D}$ be the set of states for which the mBS does not admit the arriving user, if an arrival happens, under any stationary policy $\pi$ at the given $\lambda$. The expected average cost under our threshold policy is:  
	\begin{align*}
		\rho(\lambda) &=\inf_{\pi} \Bigg[C\sum_{k} k \mu(k)+\lambda 
		\sum_{k \in \mathcal{D}}\mu(k) \Bigg]
		&=g(\lambda,t(\lambda)).
	\end{align*}
	By Lemma~\ref{Lemma 8}, $g$ is submodular; hence, by Lemma~\ref{Lemma 7}, the threshold $t(\lambda)$ is a non-decreasing function of $\lambda$. The set of states $\mathcal{D}$, for which the mBS does not admit the arriving user if any, under the threshold stationary policy is of the form $[t(\lambda),\infty)$. Therefore, $\mathcal{D}$ monotonically decreases from the whole state space to the empty set as $\lambda$ increases from $-\infty$ to $+ \infty$. Hence the problem is Whittle indexable.
\end{proof}

\section{Computation of Whittle index} \label{Sec 6}
There is a large body of work on Whittle index computation schemes~\cite{nino2007characterization,glazebrook2009index,nino2012admission,nino2012towards,Borkar2017AOPR}. We use a recursive approach similar to the one used in~\cite{Borkar2017AOPR} to compute the Whittle index for each state. Under this approach, given the state $x$, $\lambda$ is updated as follows:
	\begin{align}
		\lambda_{t+1}=&\lambda_{t} +\alpha \Big(\sum_{i} p_{i|x}(1) V_{\lambda_{t}}(i)-\sum_{i} p_{i|x}(0) V_{\lambda_{t}}(i)-\lambda_{t} \Big), \label{EQ:lambda:iteration} \\
		& t\ge 0, \nonumber
	\end{align}
where $\alpha>0$ and $p_{\cdot|\cdot}(1)$ (respectively, $p_{\cdot|\cdot}(0)$) is the transition probability when the mBS admits (respectively, does not admit) an arrival for the current slot. This is an incremental scheme that adjusts the current guess for the index in the direction of decreasing the discrepancy in the values of the RHS of the dynamic programming equation (\ref{Eq8}) corresponding to the two actions-- mBS admits and does not admit an arrival, which should agree for the correct value of the index. The equations for $V_{\lambda}$ form a linear system, which can be solved after each iteration of \eqref{EQ:lambda:iteration} using the current value of $\lambda$. That is, we solve the following system of equations for $V = V_{\lambda_t}$ and $\rho=\rho(\lambda_{t})$ using $\lambda=\lambda_{t}$: 
	\begin{align*}
		&V(y)=Cy-\rho+\sum_z p_{z|y}(1)V(z), \quad y\le x,\\
		&V(y)=Cy+\lambda-\rho+\sum_z p_{z|y}(1)V(z), \quad y> x,\\
		&V(0)=0.
	\end{align*}	
The value to which the iteration \eqref{EQ:lambda:iteration} converges
yields the Whittle index for a fixed state $x$. To reduce the computational cost, the above iteration is performed for a sufficiently large number of states $x$, and the Whittle indices for the remaining states are computed by interpolation.

Under the Whittle index based policy, in each time slot in which an arrival occurs, the mBS with the smallest index admits the arrival.	

\section{Application of Our Results to Sub-6 GHz Networks} \label{Section: Application of our results to sub-6 GHz networks}
To formulate the problem of user association in sub-6 GHz networks as a restless bandits problem, it is required to model the processes of user arrivals in the network, user departure from BSs, how and when actions are taken, rewards and costs involved with transition from one state to another upon action, etc. The above can be done as in Section~\ref{section 2}. Also, results similar to those in this paper can be obtained for the context of sub-6 GHz networks. Note that in~\cite{kasbekar2006online,kasbekar2006onlinepolicies}, the stability of user association policies for sub-6 GHz networks was studied. Similar techniques were used and analogous results were obtained for mmWave networks in~\cite{Gupta21Access}.

\section{Other policies for comparison}	
\label{Sec 7}
In this section, we briefly describe the load based, SNR based, throughput based and mixed policies, whose performance in mmWave networks was evaluated via analysis and simulations in~\cite{Gupta21Access}, and the random policy. In Section~\ref{Sec 8}, we compare the performance of our proposed Whittle index based policy with the above five policies via simulations. 
	
\subsection{Load based policy}
\label{subs 7A} 
Under this policy, in each slot in which an arrival occurs, the mBS with the minimum number of users in its queue at the beginning of the slot admits the arrival (ties are broken at random). That is, in slot $n$, mBS $\argmin_{i\in \{1,2,\cdots,K\}} X_n^i$ admits the arriving user if an arrival occurs.

\subsection{SNR based policy}\label{subs 7B}
Under this policy, in each slot in which an arrival occurs, the mBS that provides the highest data rate admits the arrival. That is, in slot $n$, mBS $\argmax_{i\in \{1,2,\cdots,K\}} r_i$ admits the arriving user if an arrival occurs.

\subsection{Throughput based policy}\label{subs 7C}
Under this policy, in each slot in which an arrival occurs, the mBS that provides the highest throughput upon association admits the arrival (ties are broken at random). That is, in slot $n$, mBS $\argmax_{i\in \{1,2,\cdots,K\}} \frac{r_i}{X^i_n+1}$ admits the arriving user if an arrival occurs.

\subsection{Mixed policy}\label{subs 7D} 
Under this policy, in each slot in which an arrival occurs, the mBS that has the highest weighted sum of the data rate and a positive scalar times the throughput upon association at the beginning of the slot admits the arrival (ties are broken at random).  That is, in slot $n$, mBS $\argmax_{i\in \{1,2,\cdots,K\}} \Big(0.2* r_i+ \frac{r_i}{X^i_n+1} \Big) $ admits the arriving user if an  arrival occurs. The reason for choosing the particular value, $0.2$, for the positive scalar is that it was shown to result in good performance of the policy in~\cite{kasbekar2006online}.  

\subsection{Random policy}
Under this policy, in each slot in which an arrival occurs, an mBS that is selected uniformly at random out of the $K$ mBSs admits the arrival. 
	
\section{Simulations} \label{Sec 8}
In this section, we compare the performance of the proposed Whittle index based user association policy  with those of the SNR based, load based, throughput based, mixed and random user association policies via simulations. We use the following performance metrics: long-run average per slot cost incurred at mBSs, average delay, i.e., the average difference between the time slots at which a user departs and arrives, and blocking probability, i.e., the ratio of the total number of arrivals that see full buffers at all the mBSs and are hence blocked, and the total number of arrivals, under an association policy. 

To show that our proposed Whittle index based association policy is robust in the sense that it outperforms other association policies briefly described in Section~\ref{Sec 7} in all mmWave network scenarios, we consider ten different mmWave network scenarios by varying  mmWave network parameters along three different axes-- $p$, $<K,r>$, and  $C$. Note that $r=[r_1, \cdots,r_K]$ (respectively, $C=[C_1,\cdots,C_K]$) denotes the vector of data rates (respectively, costs). Recall that the $i$'th component of $r$ (respectively, $C$) is the data rate (respectively, cost) corresponding to the $i$'th mBS of the network. We consider two mmWave network scenarios with respect to $<K,r>$-- the first with $K=5$ and $r = [0.55, 0.52, 0.50, 0.48, 0.45]$ and the second with $K=10$ and  $r = [0.75, 0.65, 0.62, 0.60, 0.55, 0.52, 0.50, 0.48, 0.45, 0.42]$.  We consider three arrival scenarios, two with fixed arrival  probabilities-- $p=0.4$ and $p=0.9$-- representing light and heavy load, respectively, and one with dynamically selected arrival probabilities-- in particular, in each slot, $p$ takes a  value that is selected uniformly at random from the range $[0.01, 0.99]$ and independently of the values in other slots. We consider two scenarios with respect to $C$-- the first increasing in its component value and the second decreasing in its component value. We assume that at the beginning of the simulations, the network is idle, i.e., the initial state profile is the all zeros vector, in all the considered mmWave network scenarios. The time horizon, $T$, used for the simulations is $20000$ slots. We plot the performance of the  six policies in terms of the average cost in the last $10000$ slots only; this is because we are interested in the long-run average cost.

In Figs.~\ref{fig2}-\ref{fig6}, we have plotted the long-run average cost under the six association policies versus time for different parameter values. It can be seen that \emph{in all the plots, the Whittle index based association policy outperforms all the other association policies}. Note that although the SNR  based user association policy has been extensively used in wireless  networks in practice, in Figs.~\ref{fig2}-\ref{fig3},~\ref{fig4b},~\ref{fig5} and~\ref{fig6}, it performs significantly worse than the  Whittle index based, load based, throughput based and mixed policies; also, in  Figs.~\ref{fig2b},~\ref{fig3},~\ref{fig5} and~\ref{fig6}, it performs even  worse than the random policy. The reason is that the SNR based user association policy ignores the numbers of users currently associated with different mBSs and hence leads to load imbalance.

Based on the observations from simulations in different mmWave network scenarios, we can provide the following conclusion about the impact on the performance of our proposed association policy as we vary mmWave network parameters. For a given $K,r, \mbox{ and } C$, as the arrival probability of the mmWave network varies from low load corresponding to a low value of $p$ to a high load corresponding to a high value of $p$, the performance gap between our proposed association policy and its closest competitor association policy starts decreasing.

To compare the performance of the association policies in terms of the average delay and the blocking probability, we vary the number of mBSs, $K$, from $2$ to $6$. For each value of $K$, $p$ and the sum of the components of $r$ are the same and equal $0.8$. We use the vectors of data rates: $r = [0.6, 0.2]$, $[0.4, 0.2667, 0.1333]$, $[0.3, 0.2333, 0.1667, 0.1]$, $[0.24, 0.2, 0.16, 0.12, 0.08]$, and $[0.2, 0.1733, 0.1467, 0.12, 0.0933, 0.0667]$ for the mmWave networks with $K=2$, $3$, $4$, $5$, and $6$, respectively. Also, we use the vectors of costs: $C = [10, 30]$, $[10, 20, 30]$, $[10, 16.67, 23.54, 30]$, $[10, 15, 20, 25, 30]$, and $[10, 14, 18, 22, 26, 30]$ for the mmWave networks with $K=2$, $3$, $4$, $5$, and $6$, respectively. The buffer size of each mBS is assumed to be $20$ for each value of $K$. For all the simulations done for studying the average delay and the blocking probability, we assume that  at the beginning of the simulations, the network is idle. The time horizon, $T$, used for the simulation of the network with $K$ mBSs is $K \times 5000$ slots. 

We have plotted the average delay and blocking probability under the six association policies versus $K$ in Figs.~\ref{fig7a} and~\ref{fig7b}, respectively. It can be seen that \emph{in terms of both average delay and blocking probability, the Whittle index based policy outperforms all the other association policies}.

\begin{figure}
	\centering
	\begin{subfigure}{.49\textwidth}
		\centering
		\includegraphics[width=1.12\linewidth]{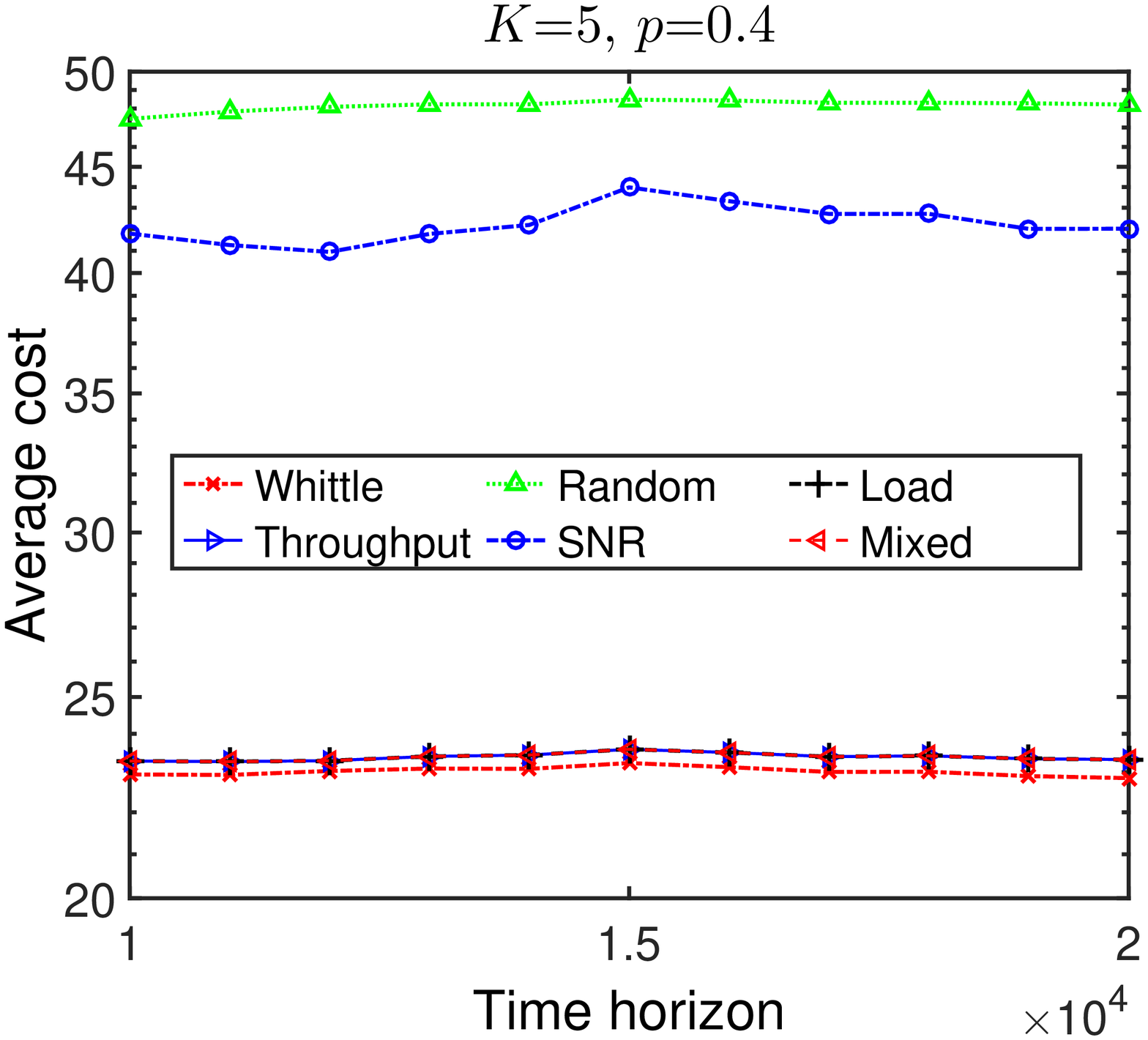}
		\caption{}
		\label{fig2a}
	\end{subfigure}
	\begin{subfigure}{.49\textwidth}
		\centering
		\includegraphics[width=1.12\linewidth]{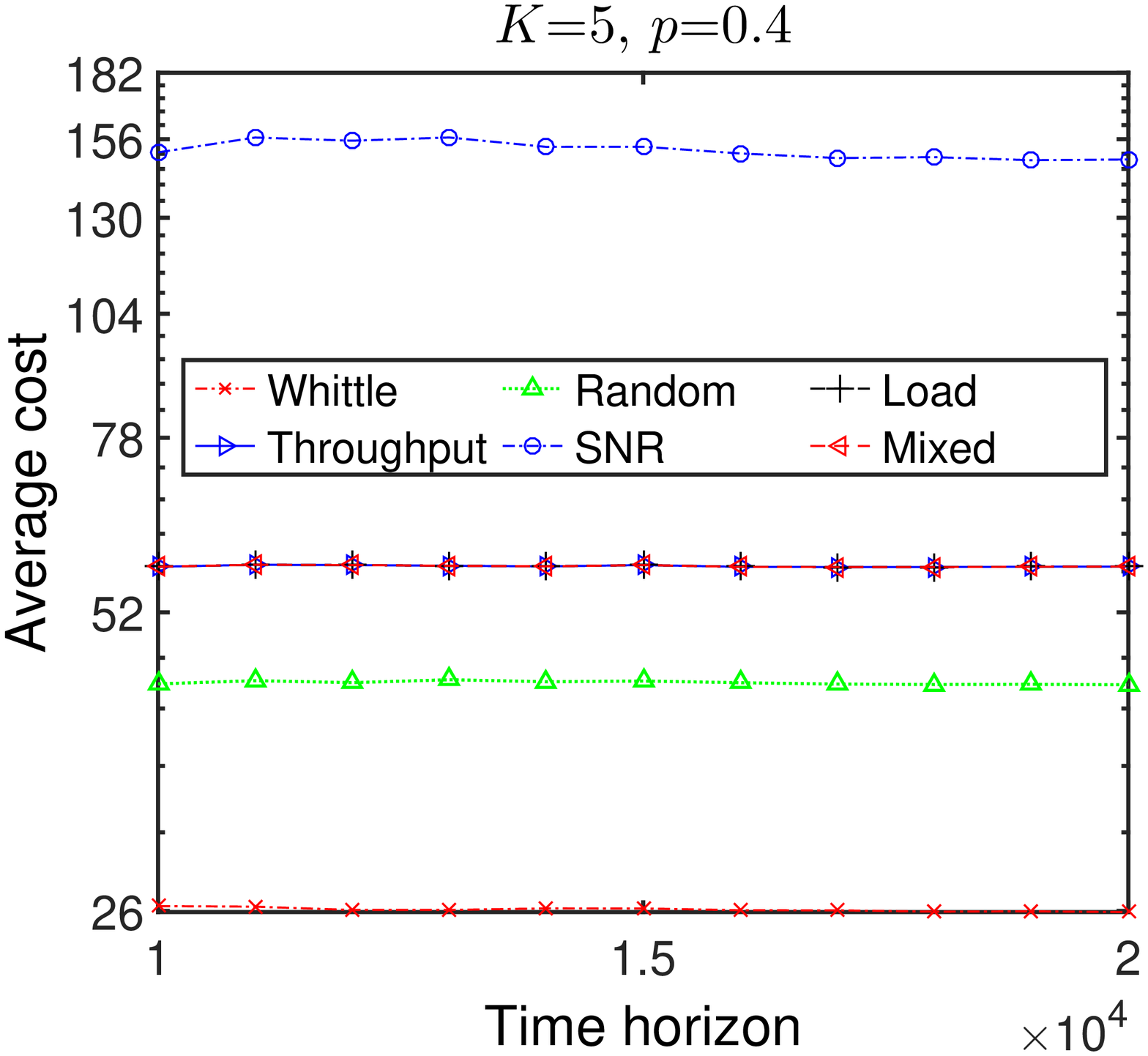}
		\caption{}
		\label{fig2b}
	\end{subfigure}
	\caption{The plots show a comparison of the average costs under the six association policies for a network with $p=0.4$ and $K = 5$. The parameter $C = [25, 35, 45, 60, 95]$ for Fig.~\ref{fig2a} and  $C = [95, 60, 45, 35, 25]$ for Fig.~\ref{fig2b}.}
	\label{fig2}
\end{figure}

\begin{figure}
	\centering
	\begin{subfigure}{.49\textwidth}
		\centering
		\includegraphics[width=1.12\linewidth]{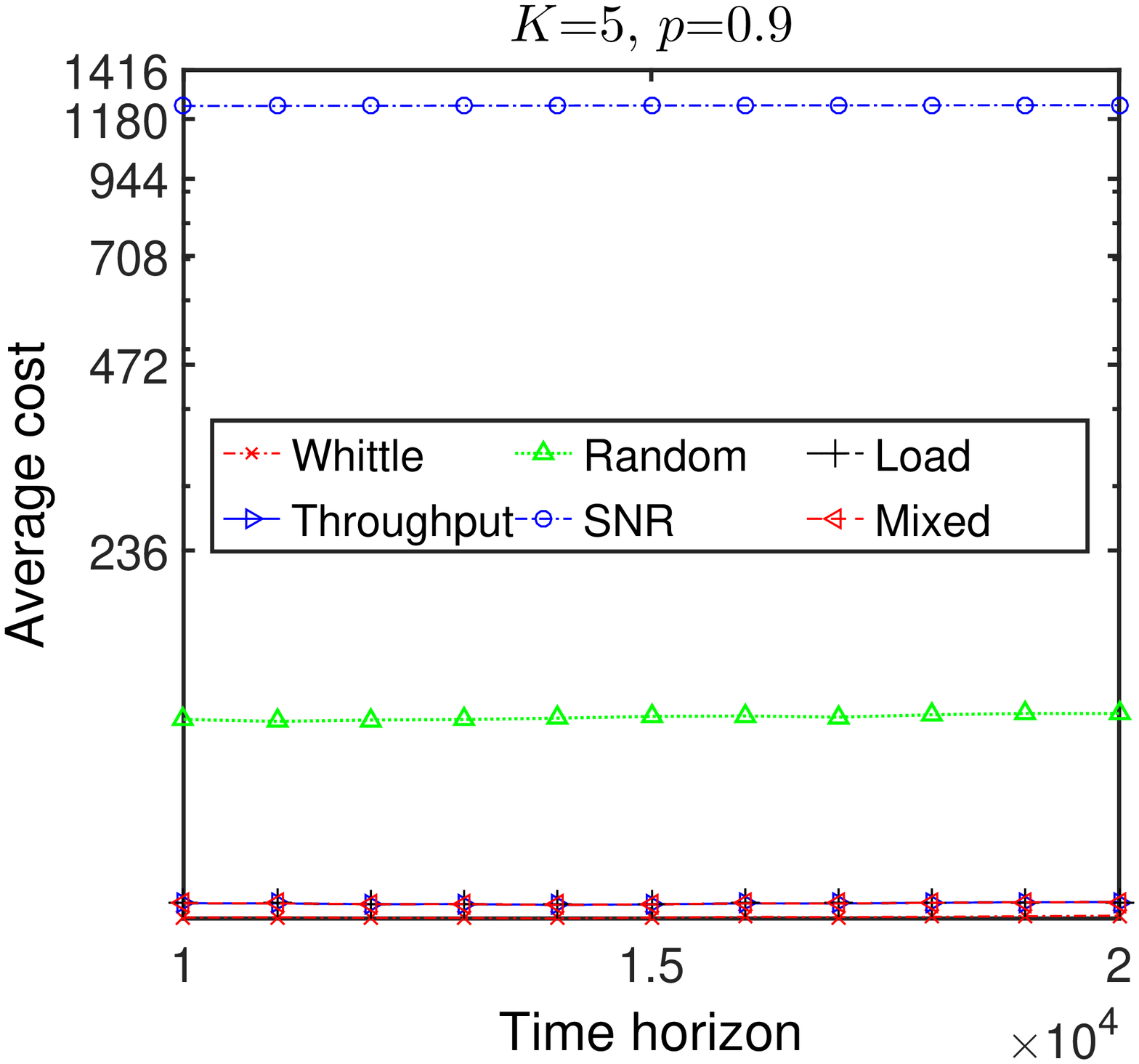}
		\caption{}
		\label{fig3a}
	\end{subfigure}
	\begin{subfigure}{.49\textwidth}
		\centering
		\includegraphics[width=1.12\linewidth]{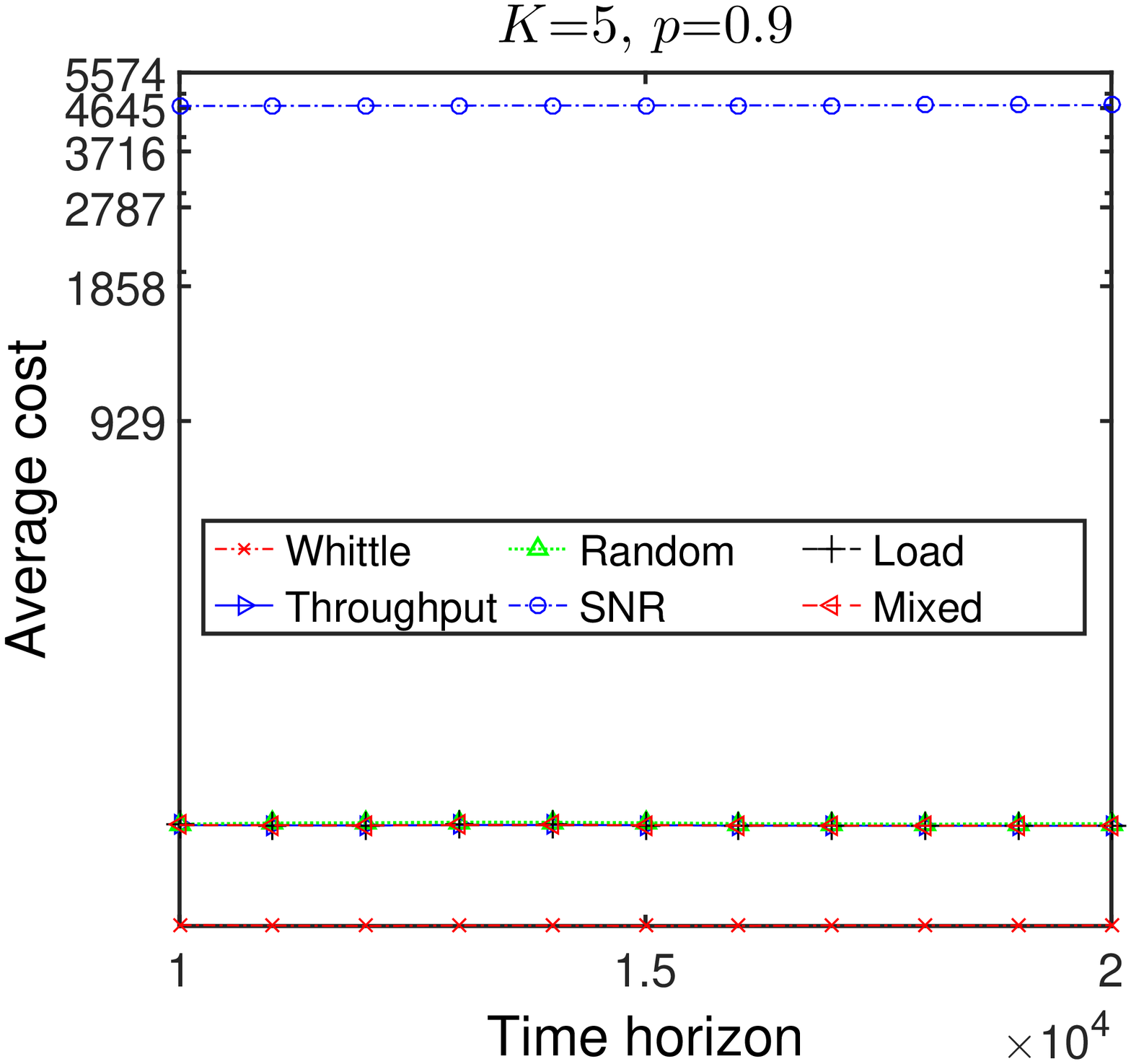}
		\caption{}
		\label{fig3b}
	\end{subfigure}
	\caption{The plots show a comparison of the average costs under the six association policies for a network with $p=0.9$ and $K = 5$. The parameter $C = [25, 35, 45, 60, 95]$ for Fig.~\ref{fig3a} and  $C = [95, 60, 45, 35, 25]$ for Fig.~\ref{fig3b}.}
	\label{fig3}
\end{figure}

\begin{figure}
	\centering
	\begin{subfigure}{.49\textwidth}
		\centering
		\includegraphics[width=1.12\linewidth]{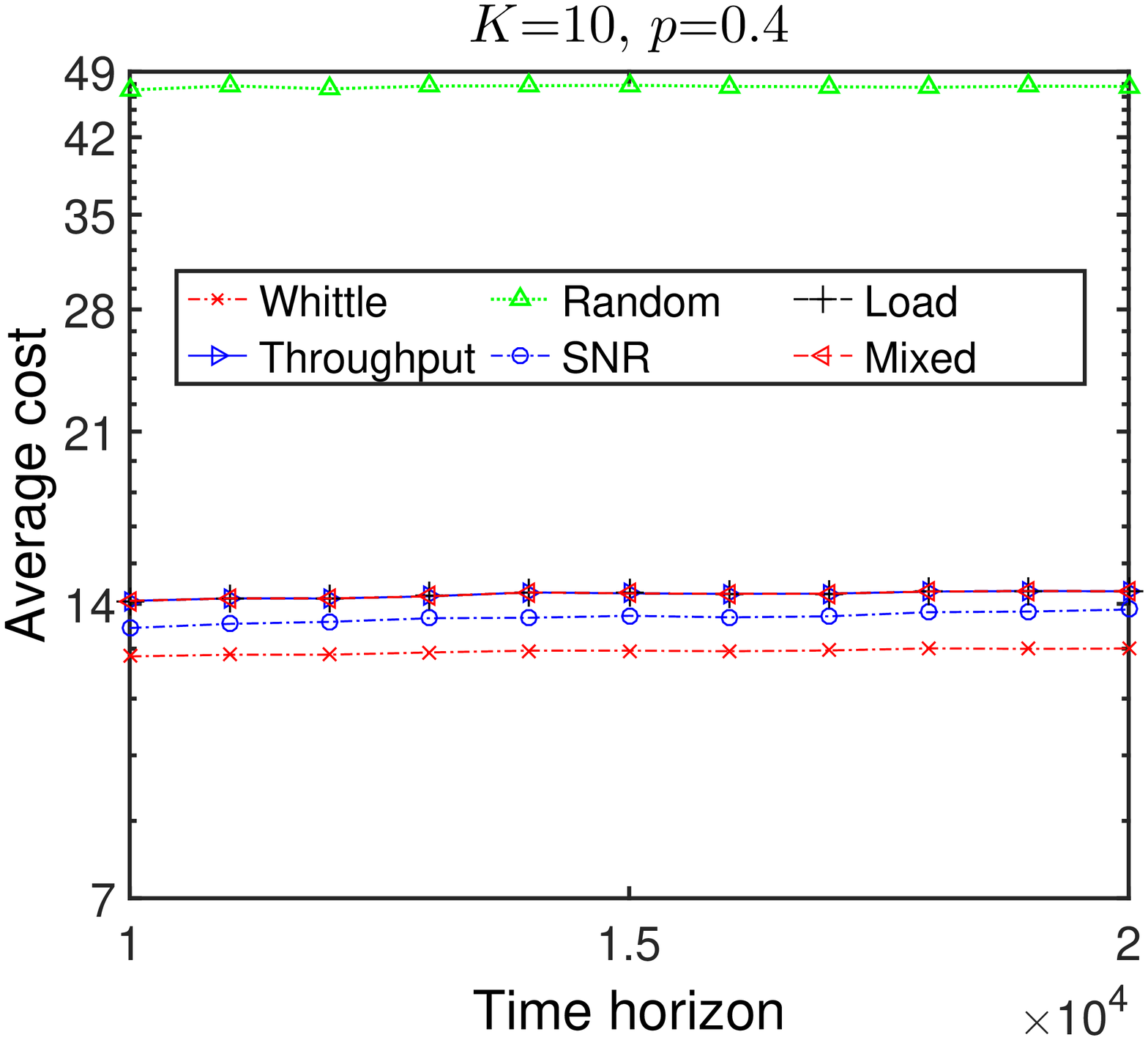}
		\caption{}
		\label{fig4a}
	\end{subfigure}
	\begin{subfigure}{.49\textwidth}
		\centering
		\includegraphics[width=1.12\linewidth]{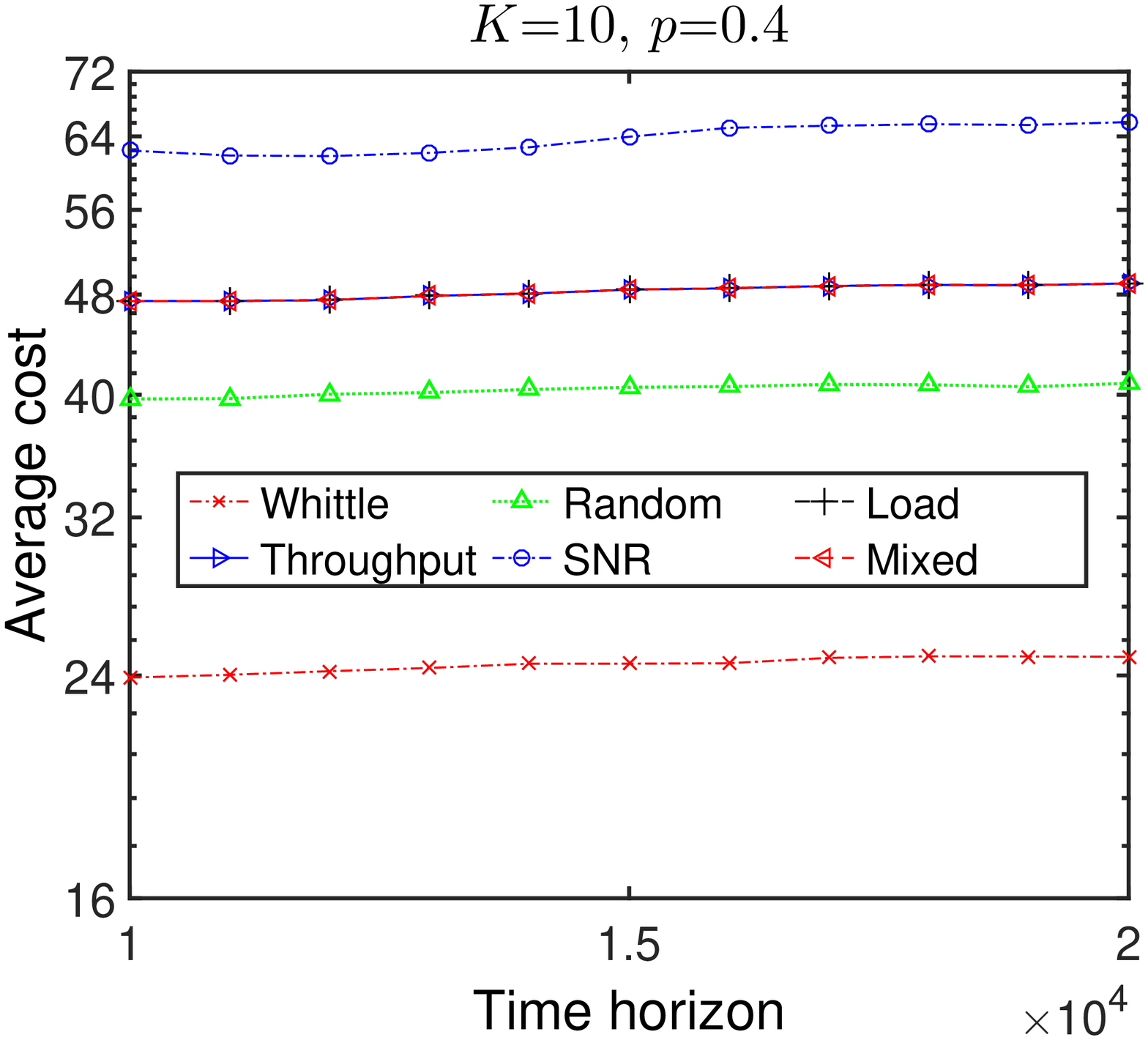}
		\caption{}
		\label{fig4b}
	\end{subfigure}
	\caption{The plots show a comparison of the average costs under the six association policies for a network with $p=0.4$ and $K = 10$. The parameter $C = [20, 32, 45, 50, 55, 60, 65, 70, 75, 95]$ for Fig.~\ref{fig4a} and $C = [95, 75, 70, 65, 60, 55, 50, 45, 32, 20]$ for Fig.~\ref{fig4b}.}
	\label{fig4}
\end{figure}

\begin{figure}
	\centering
	\begin{subfigure}{.49\textwidth}
		\centering
		\includegraphics[width=1.12\linewidth]{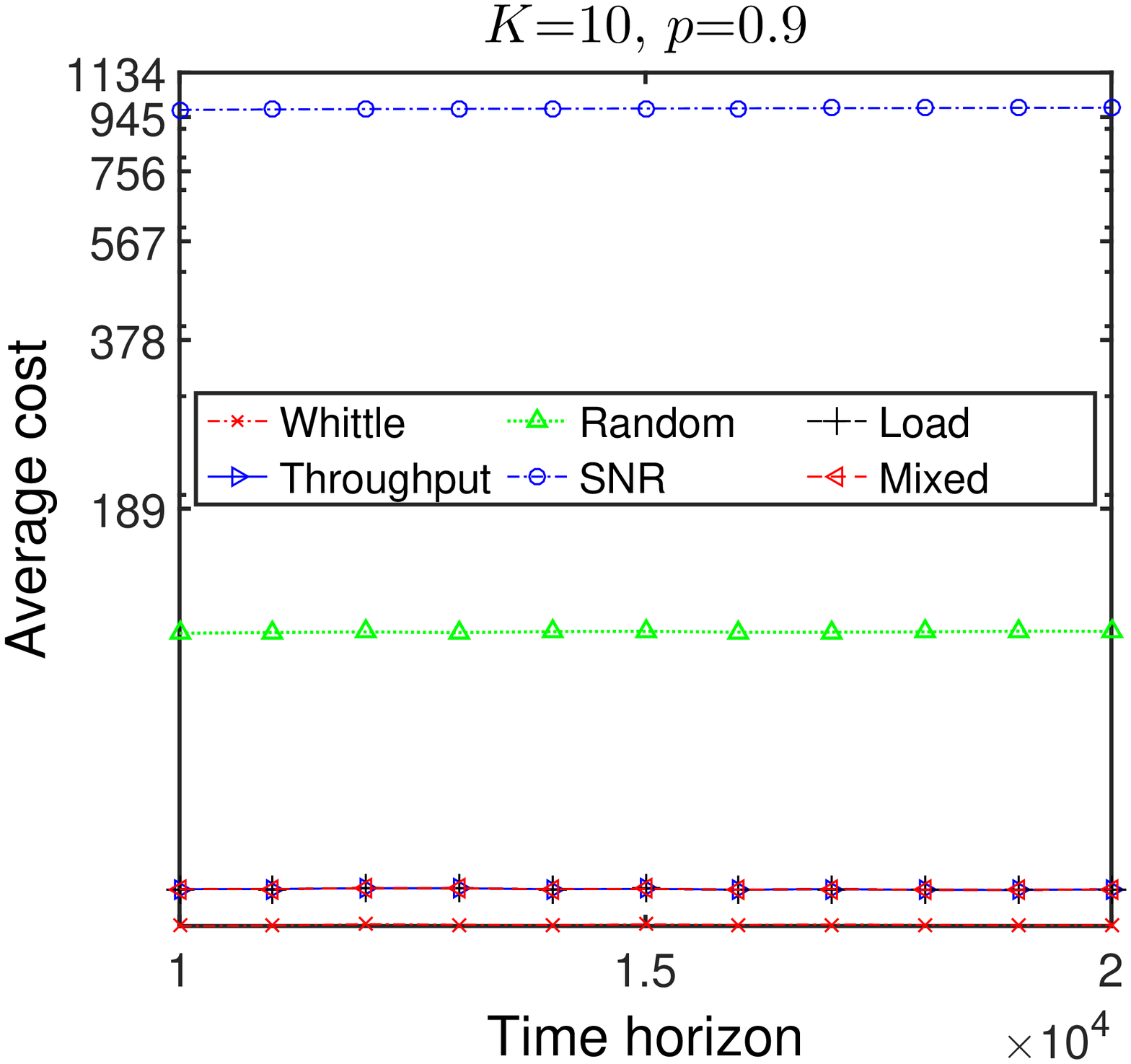}
		\caption{}
		\label{fig5a}
	\end{subfigure}
	\begin{subfigure}{.49\textwidth}
		\centering
		\includegraphics[width=1.12\linewidth]{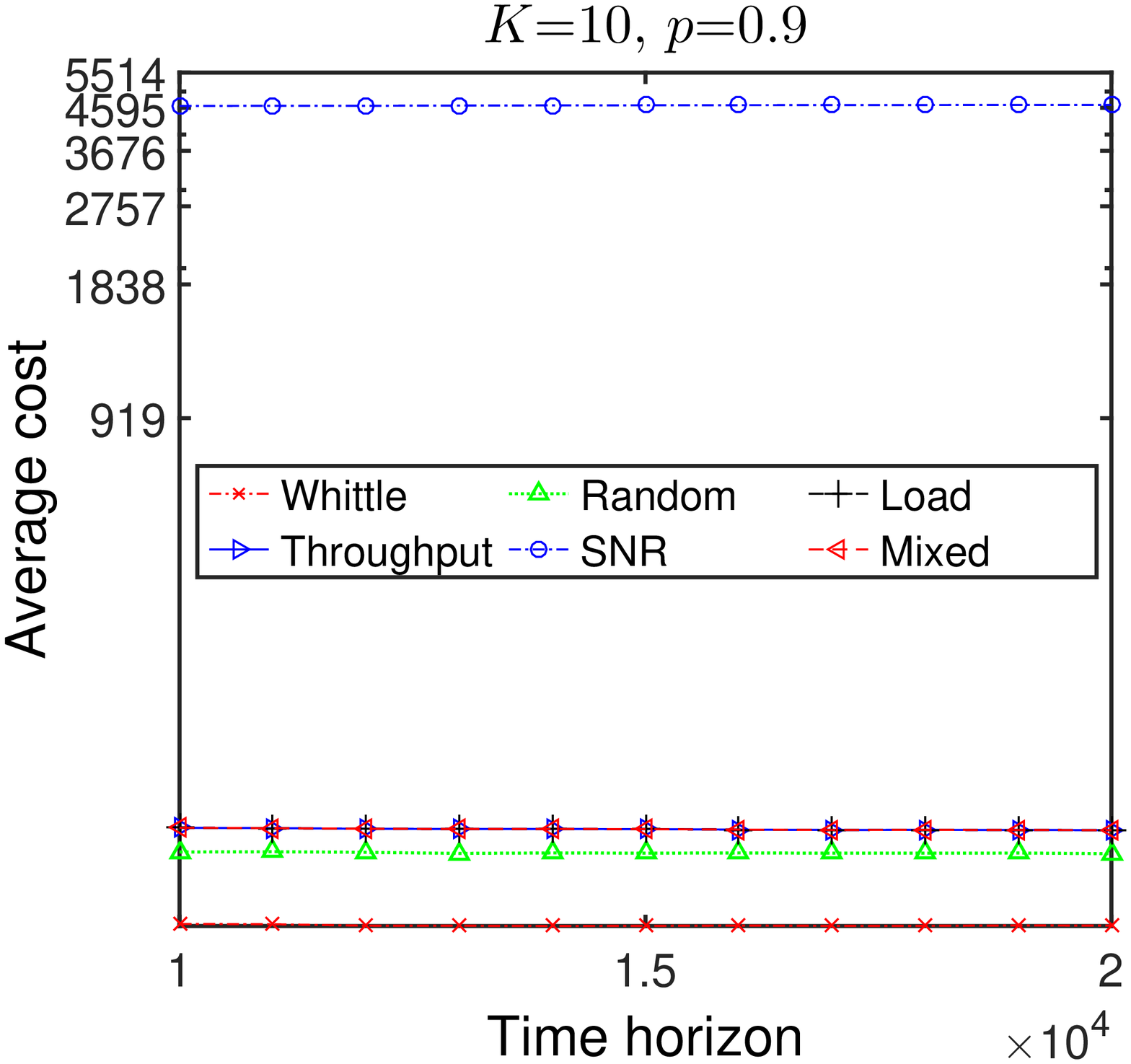}
		\caption{}
		\label{fig5b}
	\end{subfigure}
	\caption{The plots show a comparison of the average costs under the six association policies for a network with $p=0.9$ and $K = 10$. The parameter $C = [20, 32, 45, 50, 55, 60, 65, 70, 75, 95]$ for Fig.~\ref{fig5a} and $C = [95, 75, 70, 65, 60, 55, 50, 45, 32, 20]$ for Fig.~\ref{fig5b}.}
\label{fig5}
\end{figure}

\begin{figure}
	\centering
	\begin{subfigure}{.49\textwidth}
		\centering
		\includegraphics[width=1.12\linewidth]{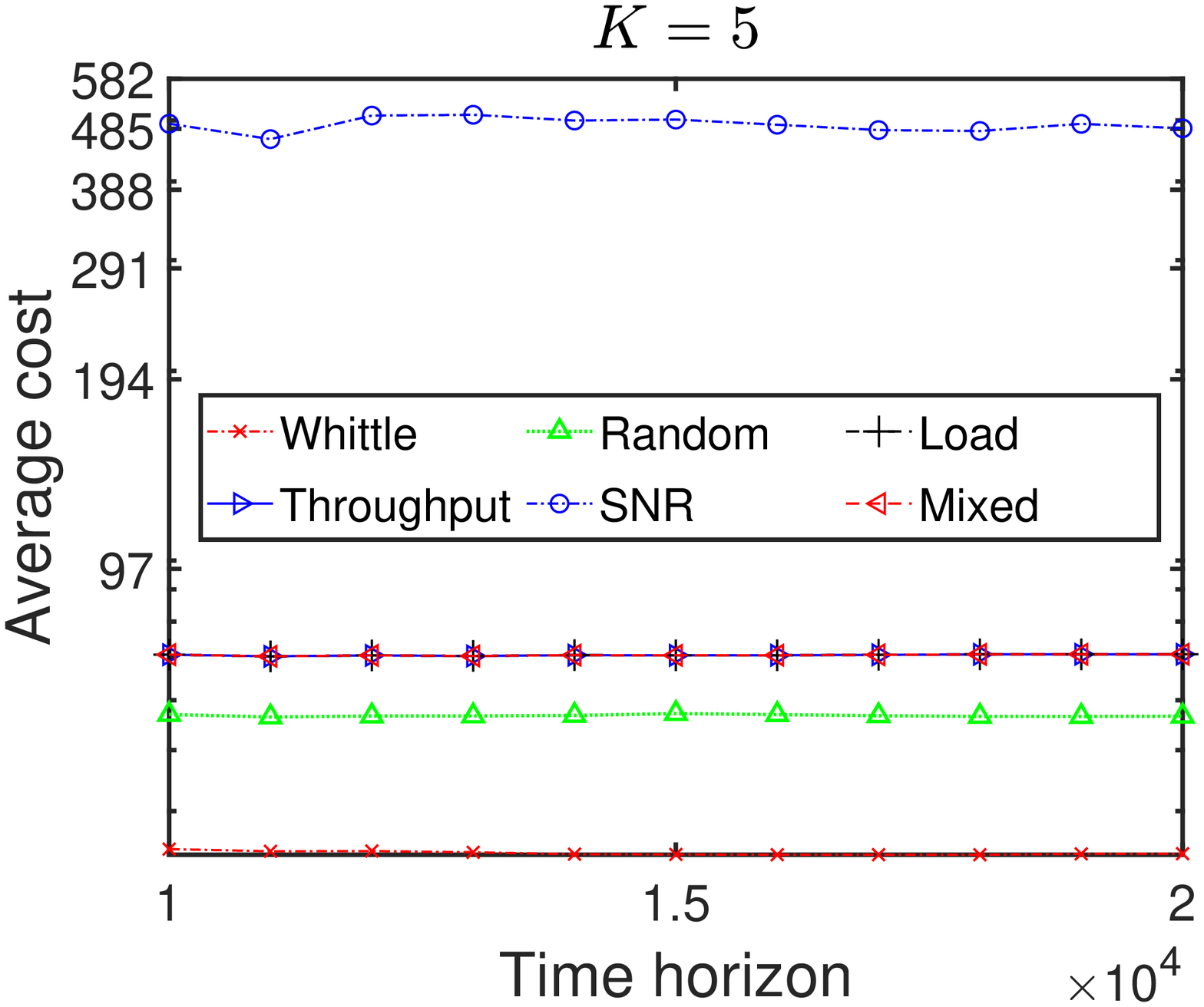}
		\caption{}
		\label{fig6a}
	\end{subfigure}
	\begin{subfigure}{.49\textwidth}
		\centering
		\includegraphics[width=1.12\linewidth]{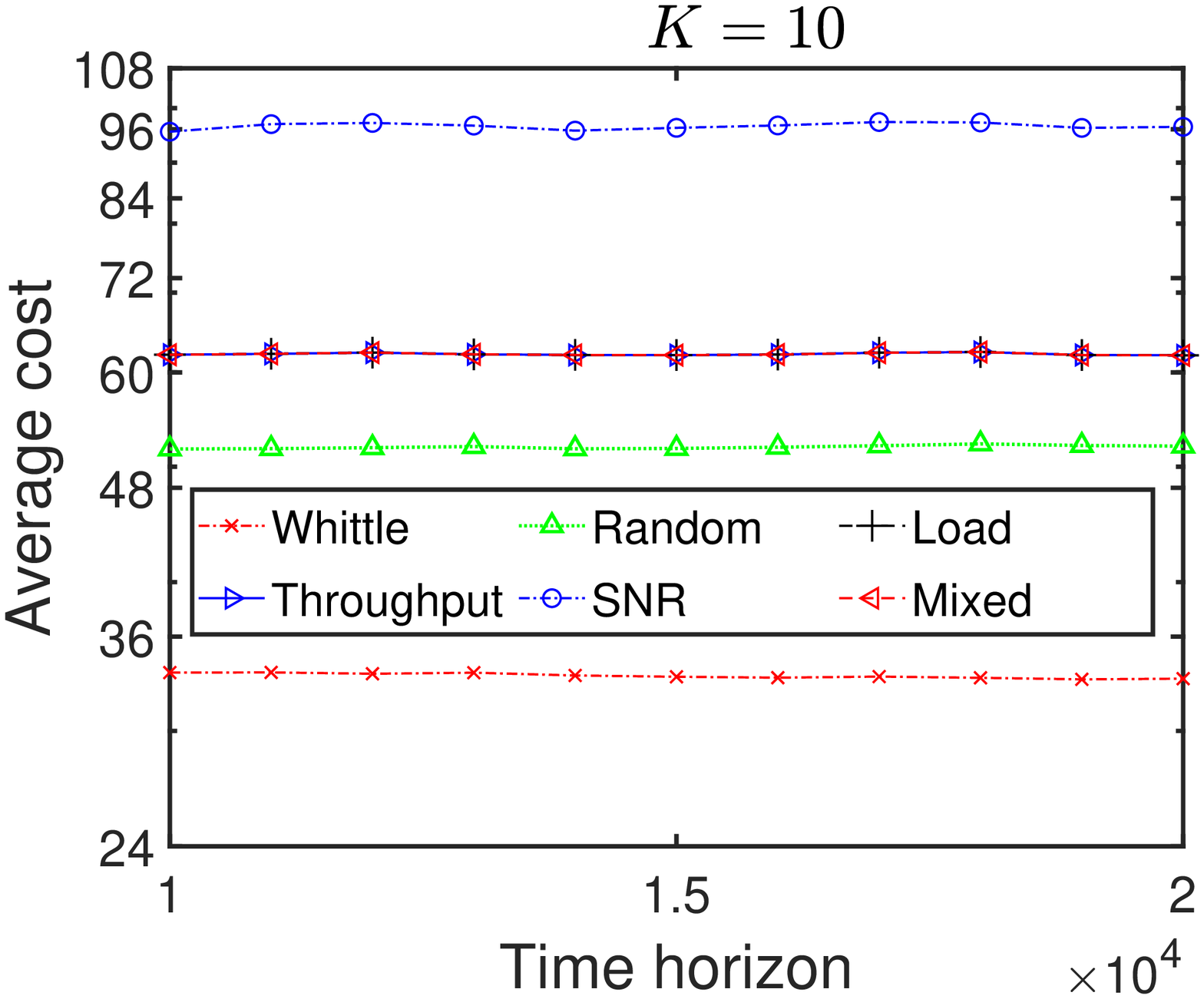}
		\caption{}
		\label{fig6b}
	\end{subfigure}
	\caption{The plots show a comparison of the average costs under the six association policies for a network with dynamically selected $p$. The parameters for Fig.~\ref{fig6a} (respectively, Fig.~\ref{fig6b}) are $K=5$ and $C = [95, 60, 45, 35, 25]$  (respectively, $K=10$ and $C = [95, 75, 70, 65, 60, 55, 50, 45, 32, 20]$).}
	\label{fig6}
\end{figure}

\begin{figure}
	\centering
	\begin{subfigure}{.49\textwidth}
		\centering
		\includegraphics[width=1.12\linewidth]{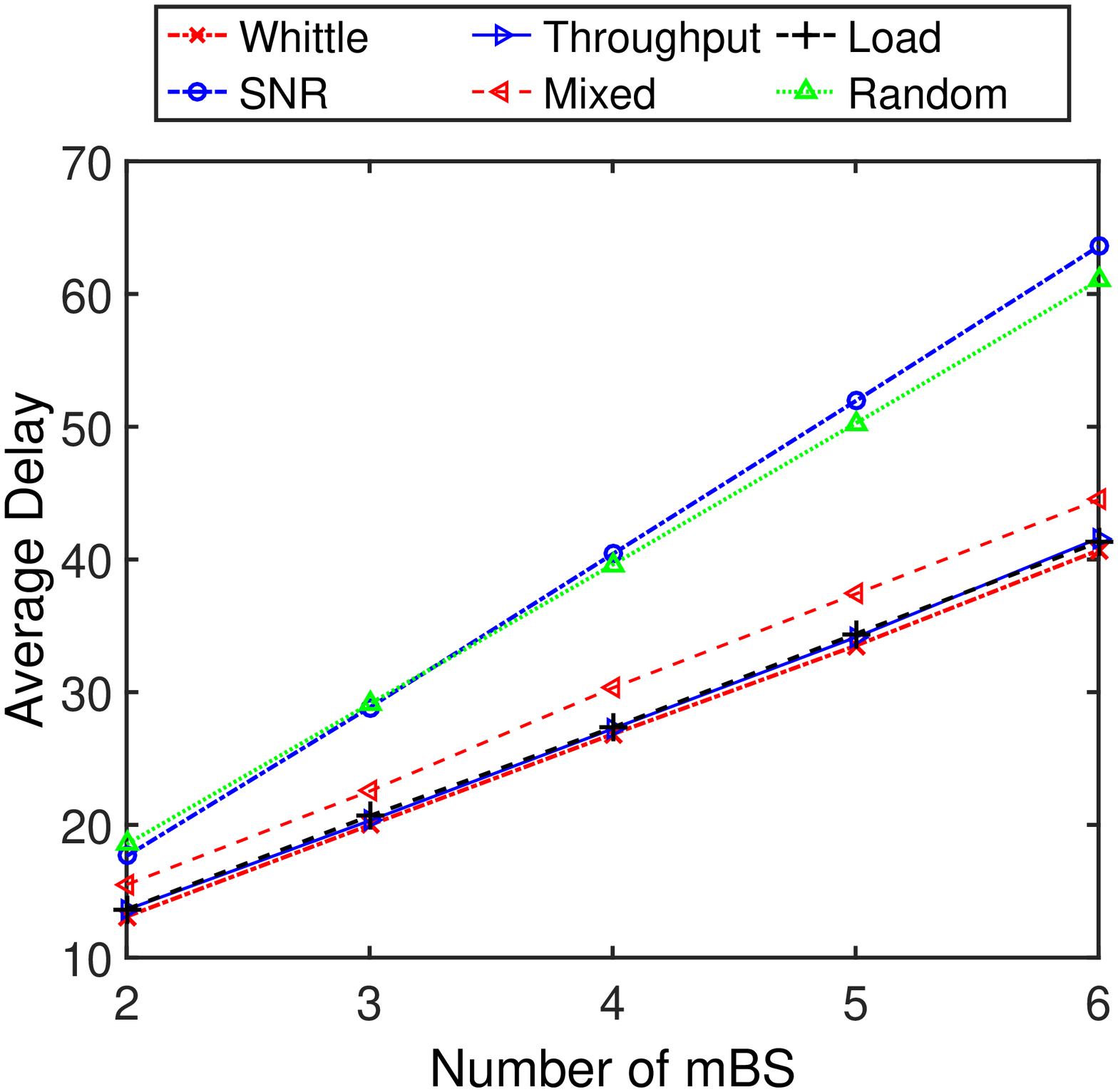}
		\caption{}
		\label{fig7a}
	\end{subfigure}
	\begin{subfigure}{.49\textwidth}
		\centering
		\includegraphics[width=1.12\linewidth]{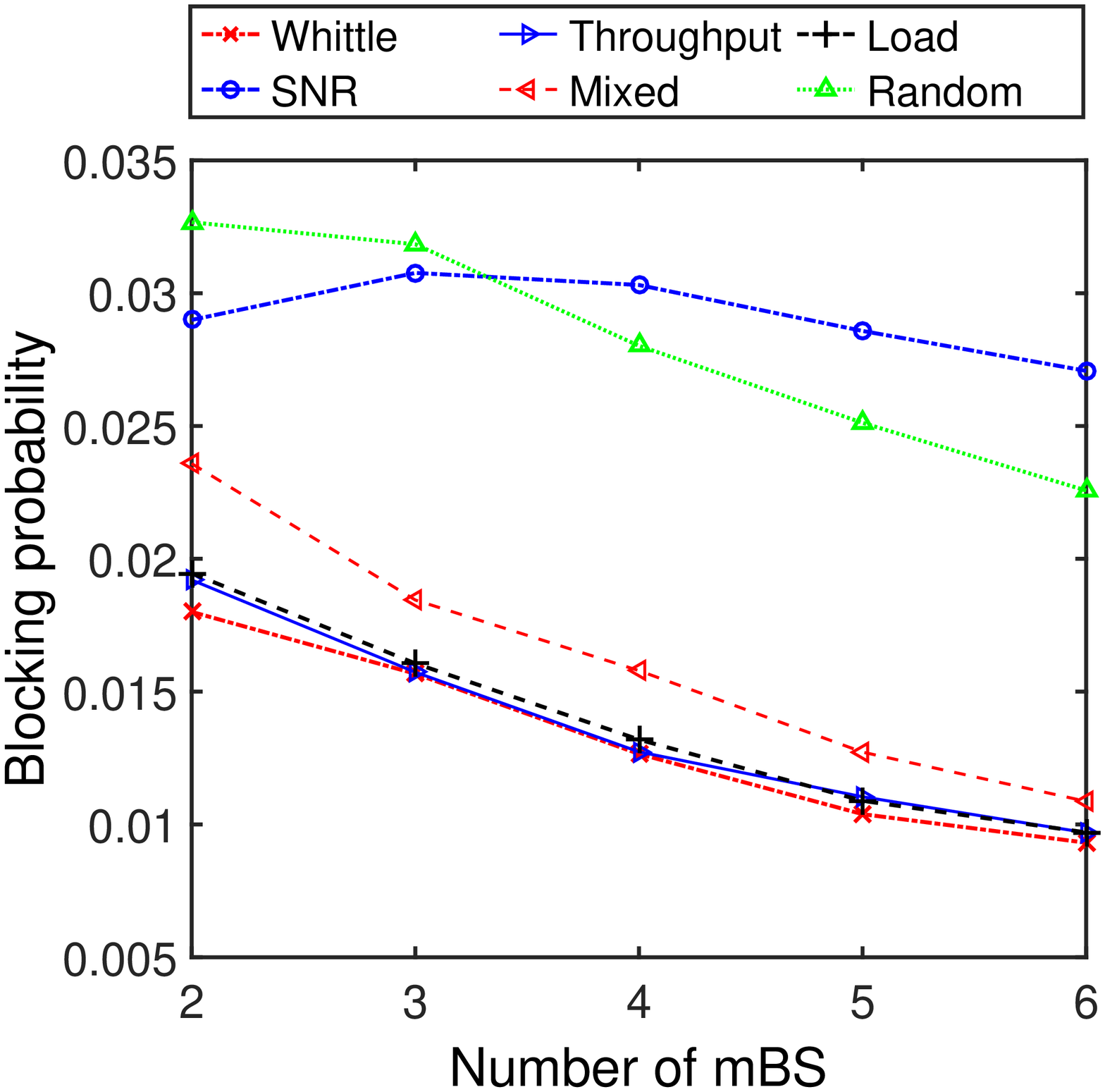}
		\caption{}
		\label{fig7b}
	\end{subfigure}
	\caption{Fig.~\ref{fig7a} (respectively, Fig.~\ref{fig7b}) shows a comparison of the average delay (respectively, blocking probability) under the six association policies for different values of $K$.}
	\label{fig7}
\end{figure}

\section{Conclusions}
\label{SC:conclusions}
In this paper, we studied the problem of user association, i.e., determining which BS an arriving user should associate with, in a dense mmWave network.  
We formulated this as a restless multi-armed bandit problem, which is provably hard to solve. We established the Whittle indexability of the problem, and based on this result, devised an association policy, in which an arriving user  associates with the BS that has the smallest Whittle index. Using simulations, we showed that our Whittle index based association policy outperforms the  SNR based, throughput based, load based and mixed policies proposed in prior work. 

	\bibliography{library} 
	\bibliographystyle{IEEEtran}
\vspace{1em}
\begin{wrapfigure}{l}{25mm} 
	\includegraphics[width=1in,height=1in,clip,keepaspectratio]{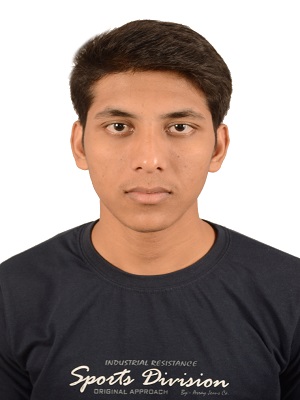}
\end{wrapfigure}\par
\textbf{Santosh Kumar Singh} received the B.Tech. degree in electronics and communication (EC) engineering from the Indian Institute of Technology (IIT) at Roorkee, Roorkee, India, in 2015, the M.Tech. degree in communication engineering from the Indian Institute of Technology (IIT) at Delhi, Delhi, India, in 2017. He is currently pursuing the Ph.D. degree with the Department of Electrical Engineering, Indian Institute of Technology (IIT) at Bombay, Mumbai, India. His research interests include modeling, design, and analysis of resource allocation algorithms in mmWave networks.

\vspace{1em}
\begin{wrapfigure}{l}{25mm} 
	\includegraphics[width=1in,height=1in,clip,keepaspectratio]{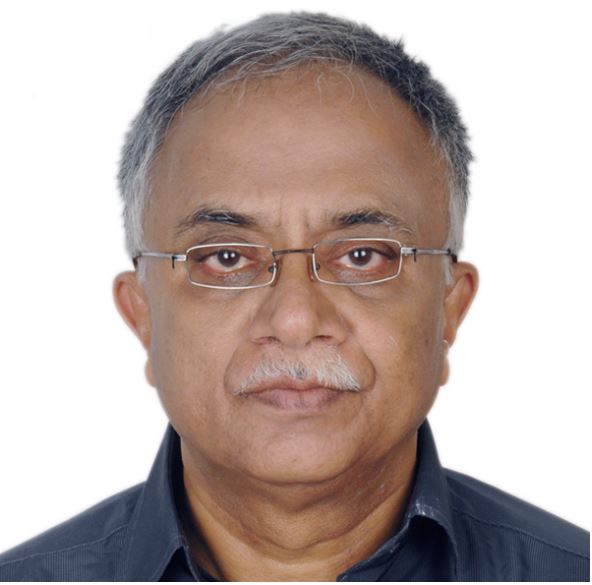}
\end{wrapfigure}\par
\textbf{Vivek S. Borkar} received B.Tech. in Electrical Engineering from IIT Bombay in 1976, M.S. in Systems and Control Engineering from Case Western Reserve University, Cleveland, in 1977, and Ph.D. in Electrical Engineering and Computer Science from the University of California at Berkeley in 1980. He has held regular positions at the TIFR Centre for Applicable Mathematics and the Indian Institute of Science in Bengaluru and the Tata Institute of Fundamental Research and Indian Institute of Technology Bombay in Mumbai. He is currently an Emeritus Fellow in the latter. He is a Fellow of the IEEE, the American Mathematical Society, The World Academy of Sciences, and various science and engineering academies in India. He has won many national honors including the S.S. Bhatnagar Award, the Prasanta Chandra Mahalanobis Medal of the Indian National Science Academy, and the Homi Bhabha and S. S. Bhatnagar Fellowships. His research interests include control of Markov processes, stochastic approximation algorithms and reinforcement learning.

\vspace{1em}
\begin{wrapfigure}{l}{25mm} 
	\includegraphics[width=1in,height=1in,clip,keepaspectratio]{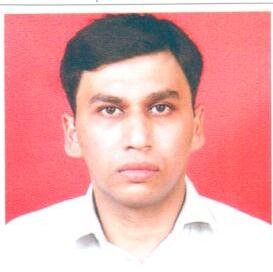}
\end{wrapfigure}\par
\textbf{Gaurav S. Kasbekar} received B.Tech. in Electrical
Engg. from Indian Institute of Technology (IIT), Bombay
in 2004, M.Tech. in Electronics Design and
Technology (EDT) from Indian Institute of Science
(IISc), Bangalore in 2006 and Ph.D from University of Pennsylvania, USA in 2011. He is currently an Associate Professor with the
Department of Electrical Engineering, IIT Bombay. His research interests
are in communication networking and network security. He received the CEDT Design Medal for being adjudged the best Masters student in EDT at IISc.\par
\end{document}